\documentclass[9pt]{amsart}
\usepackage{amsmath}
\usepackage{amssymb}
\usepackage{epsf}
\usepackage{subfigure}
\usepackage{amsfonts}
\usepackage{graphics}
\usepackage{graphicx}
\usepackage{enumitem}
\usepackage{booktabs}
\usepackage{xcolor}
\usepackage{float}

\usepackage[linesnumbered,lined,boxed,commentsnumbered]{algorithm2e} 
\usepackage[noend]{algpseudocode}

\DeclareMathOperator{\dom}{\mathrm{dom}}
\newcommand{\Diag}[1]{\mathrm{Diag}\left(#1\right)}
\newcommand{\diag}[1]{\mathrm{diag}\left(#1\right)}
\newcommand{\conv}{\mathrm{conv}}
\newcommand{\clco}{\overline{\conv}}
\newcommand{\gammax}{{\gam_{\mbox{\tiny max}}}}
\def\prox{{\mathrm{prox}}}

\DeclareMathOperator*{\argmin}{\mathrm{argmin}}
\newcommand{\intr}[1]{\mathrm{int}\left( #1\right)}

\DeclareMathOperator{\B}{\mathbb{B}}
\DeclareMathOperator{\R}{\mathbb{R}}

\DeclareMathOperator{\N}{\mathbb{N}}

\DeclareMathOperator{\bS}{\mathbb{S}}

\DeclareMathOperator{\NN}{\mathcal{N}}
\def\CC{\mathcal{C}}

\DeclareMathOperator{\LL}{\mathcal{L}}

\def\SS{\mathcal{S}}
\def\hSS{{\widehat\SS}}
\DeclareMathOperator{\DD}{\mathcal{D}}

\def\uball{\mathbb{B}}

\newcommand{\tbeta}{\tilde\beta}
\newcommand{\tgamma}{\tilde\gamma}

\newcommand{\tr}{\tilde{r}}

\newcommand{\hbeta}{\hat{\beta}}

\newcommand{\bgam}{\bar{\gamma}}

\newcommand{\half}{\mbox{\small$\frac{1}{2}$}}

\newcommand\reals{{{\rm l} \kern -.15em {\rm R} }}

\DeclareMathOperator{\epi}{epi}

\newenvironment{mat}{\left[\begin{array}{ccccccccccccccc}}{\end{array}\right]}
\newcommand\bcm{\begin{mat}}
\newcommand\ecm{\end{mat}}

\newenvironment{rmat}{\left[\begin{array}{rrrrrrrrrrrrr}}{\end{array}\right]}
\newcommand\brm{\begin{rmat}}
\newcommand\erm{\end{rmat}}

\newcommand{\eq}[1]{\begin{equation}\begin{split}#1\end{split}\end{equation}}

\newtheorem{theorem}{Theorem}

\newtheorem{lemma}[theorem]{Lemma}

\def\bW{{\mathbb{W}}}

\def\eR{{\overline\R}}
\def\Rp{{\R^p}}
\def\Rq{{\R^q}}

\def\hR{{\widehat R}}
\def\tR{{\widetilde R}}

\def\bv{{\bar v}}

\def\tw{{\tilde w}}

\def\hZ{{\widehat Z}}

\def\one{\mathbf{1}}

\def\lam{\lambda}
\def\blam{{\bar\eta}}

\def\Lam{\Lambda}
\def\alf{\alpha}
\def\talf{{\tilde\alpha}}
\def\tkappa{{\tilde\kappa}}
\def\balf{{\bar\alpha}}
\def\bgam{{\bar\gamma}}
\def\bbeta{{\bar\beta}}
\def\bbeta{{\bar\beta}}
\def\bxi{{\bar\xi}}
\def\bw{{\bar w}}
\def\hw{{\hat w}}
\def\bx{{\bar x}}

\def\bv{{\bar v}}
\def\hv{{\hat v}}
\def\hx{{\hat x}}

\def\tv{{\tilde v}}
\def\gam{\gamma}
\def\hgam{{\hat\gam}}
\def\tgam{{\tilde\gam}}
\def\bgam{{\bar\gam}}
\def\Gam{\Gamma}

\def\bomega{{\bar\omega}}
\def\del{\delta}

\def\eps{\epsilon}
\def\beps{{\bar\epsilon}}

\def\sig{\sigma}

\def\VV{{\mathcal{V}}}

\def\tLL{{\widetilde\LL}}

\def\Dfr{\mathfrak{D}}

\def\hOmega{{\widehat\Omega}}
\def\mmax{\text{max}}
\def\mmin{\text{min}}
\newcommand\map[3]{#1:#2\rightarrow #3}

\newcommand{\lset}[2]{\left\{#1\,\left|\,#2\right.\right\}}
\newcommand{\norm}[1]{\left\Vert #1\right\Vert}
\newcommand{\im}[1]{{\mathrm{Im}\left(#1\right)}}

\newcommand{\ip}[2]{\left\langle #1,\, #2\right\rangle}

\def\bgk{(\beta^k,\gam^k)}
\def\tbgk{(\tbeta^k,\tgam^k)}
\def\bg{(\beta,\gam)}
\def\tbg{(\tbeta,\tgam)}
\def\bbg{(\bar\beta,\bar\gam)}
\def\hbg{(\hat\beta,\hat\gam)}
\def\emu{{\eta,\mu}}
\def\phimu{\phi_\mu}

\def\uemu{{u_{\eta,\mu}}}
\def\keta{\kappa_\eta}
\def\Rqp{\R^q_+}
\def\vphi{\varphi}
\def\vmu{v_\mu}
\def\xmu{x_\mu}

\begin{document}
\title[Feature Selection for Linear Mixed Effects]{Analysis of Relaxation Methods for Feature Selection in Mixed Effects Models}
\author{A. Aravkin}
\thanks{A.~Aravkin, Department of Applied Mathematics \&
    Institute of Health Metrics and Evaluation, University of Washington.}
\author{J.~V.~Burke}
\thanks{J.~V.~Burke, Math. Dept., University of Washington.
Supported in part by the U.S. NSF grant DMS-1514559.}
\author{A. Sholokhov}
\thanks{A.~Sholokhov,
Bill and Melinda Gates Foundation,
Department of Applied Mathematics, University of Washington}
\author{P. Zheng}
\thanks{P.~Zheng, Department of Health Metrics Sciences \&
    Institute of Health Metrics and Evaluation,
    University of Washington.
    }
\date{\today}

\maketitle

\begin{abstract}
Linear Mixed-Effects (LME) models are a fundamental tool for modeling clustered data, including cohort studies, longitudinal data analysis, and meta-analysis.  
The design and analysis of variable selection methods for LMEs is 
considerably more difficult than for linear regression because LME models are nonlinear. 
The approach considered here is motivated by a recent method for sparse relaxed regularized regression (SR3) for variable selection in the context of linear regression.
The theoretical underpinnings 
for the proposed extension to LMEs are developed, including consistency results, variational properties, 
implementability of optimization methods, and convergence results.  
In particular we provide convergence analyses for a basic implementation of SR3 for LME (called MSR3) and an accelerated 
hybred algorithm (called MSR3-fast).  
Numerical results show the utility and speed of these algorithms on realistic simulated datasets. 
The numerical implementations are available  
in an open source \texttt{python} package \texttt{pysr3}.
\end{abstract}

\noindent%

\section{Introduction}
Linear mixed-effects (LME) models are used to analyze nested 
or combined data across a range of groups or clusters.  
These models use covariates to separate the total population
variability (the fixed effects) from the group variability (the random effects).
LMEs borrow strength across groups to estimate key statistics in cases where the data within
groups may be sparse or highly variable, and 
play a fundamental role in population health sciences \cite{covid2020modeling,murray2020global}, meta-analysis \cite{dersimonian1986meta, zheng2021trimmed}, life sciences, and 
in many others domains \cite{zuur2009mixed}. 

This paper develops the theoretical bases for the algorithmic approach to variable selection within LME context presented in \cite{Practice}.
Although there are many successful algorithms and software
for variable selection
for linear regression, e.g. 
the Lasso method~(\cite{tibshirani1996regression,glmnet}) and related extensions,
approaches to variable selection for LMEs are far less settled 
with few open source software tools available
despite this being an active research area for over 20 years 
\cite{Buscemi2019Survey}.

Variable selection for LMEs is significantly complicated by the underlying
nonlinear structure associated with estimating variance parameters induced by the group structure. 
In the context of 
this paper, the nonlinearities
come from the logs of the determinants of the within group variances
as well as the regularizers used for variable selection.
Current approaches make key design decisions including 
the choice of likelihood (e.g. marginal/restricted/h- likelihood),  regularizer~(e.g. $\ell_1$~\cite{Krishna2008} or SCAD~\cite{ibrahim2011fixed}), and information criteria (\cite{Vaida2005,Ibrahim2011}). 
The wide variety of these decision choices has likely contributed to small number
of standardized software tools that allow for a comparison 
of different regularizers. 
The goal in \cite{Practice} is to fill this gap by proposing a unified methodological framework that
accommodates a variety of variable selection strategies based on a range
of easily implementable regularizers. Here we provide a 
theoretical justification for the algorithmic approach to the solution
of the marginalized maximum likelihood estimation problems presented
in \cite{Practice}.

The approach is motivated by the sparse relaxed regularization regression (SR3) strategy
developed in \cite{Zheng2019SR3}. Both the approach of~\cite{Zheng2019SR3}
and the MSR3/MSR3-fast extensions for LMEs described in~\cite{Practice} use 
auxiliary variables  to decouple the smooth terms from the 
variable selection regularizer.  The original and auxiliary
variables are tethered by adding to the objective their norm squared difference.
While the analysis of~\cite{Zheng2019SR3} relies on the least squares data-fitting term,
here we develop the algorithmic design and analysis required for the 
nonlinear and nonconvex LME extension.  

The paper proceeds as follows. 
The mathematical description of the LME model is given in Section 2 along with two results on the existence of solutions and a brief description of 
a naive PGD algorithm for their solution.
The proposed relaxation strategy described in~\cite{Practice} is given in Section~\eqref{sec:relax} where the the relaxation depends on a decoupling
parameter $\eta$ and log-barrier smoothing parameter $\mu$. In particular, we introduce
the optimal value function $\uemu$ 
used to exploit the decoupling of the likelihood
function from the sparsity regularizer. Here $\uemu$ is obtained by partially minimizing
the decoupled variables in the likelihood while keeping those in the regularizer fixed. 
Section~\ref{sec:pgd} introduces the MSR3 algorithm
as the PGD algorithm 
applied to the sum of $\uemu$ and the sparsity regularizer. 
A brief discussion the basic assumptions typically required
for establishing the viability of the PGD algorithm for this formulation is given.
Section~\ref{sec:smoothness} is the theoretical core of the paper. In this section
w show that the optimal value function function $\uemu$ satisfies the properties
necessary for the application of the PGD algorithm. 
In particular, we establish the Lipschitz continuity of $\nabla \uemu$ 
(Lemma \ref{lem:grad u lip}). 
The convergence results for MSR3 are presented in Section~\ref{sec:convergence pgd}
for fixed values of $\eta$ and $\mu$.
In Section~\ref{sec:hybrid alg} we address the key issues surrounding the initialization
of the coupling and smoothing parameters $\eta$ and $\mu$ when only approximate
values for $\uemu$ and $\nabla \uemu$ are known. Here we appeal to both
variable metric ideas as well as properties of the interior point algorithm.
In Section~\ref{sec:numerical} we give a briefly 
synopsis of some of these results obtained in \cite{Practice}.
These results indicate that
the use of the optimal value function can dramatically improve both the efficiency and the
performance of the numerical solution procedure in terms of computational 
speed and accuracy in variable selection.

\section{Models and Notation}
\label{sec:models}
Consider $m$ groups of observations indexed by $i$, with sizes $n_i$, 
so that the total number of observations is $n := n_1 + n_2 + \dots + n_m$. 
Each group is paired with a design matrix of fixed features $X_i \in \R^{n_i \times p}$ and a matrix of random features $Z_i \in \R^{n_i \times q}$ along with vectors of outcomes $Y_i \in \R^{n_i}$. 
Set 
{$X := [X_1^T, X_2^T, \dots, X_m^T]^T$ and $Z := [Z_1^T, Z_2^T, \dots, Z_m^T]^T$.} 
Following~\cite{Patterson1971, Pinheiro2000}, we define a Linear Mixed-Effects (LME) model as
\eq{
	\label{eq:lme_setup}
	Y_i & = X_i\beta + Z_iu_i + \varepsilon_i, \quad i = 1 \dots m \\
	u_i & \sim \NN(0, \Gamma),\quad \Gamma \in \bS_{+}^{q} \\
	\varepsilon_i & \sim \NN(0, \Lambda_i), \quad \Lambda_i \in \bS_{++}^{n_i}
}
 where $\beta \in \R^p$ is a vector of fixed (mean) covariates, 
 $u_i \in \R^{q}$ are unobservable random effects assumed to be distributed normally with zero mean and the unknown covariance matrix $\Gamma$, with $\bS_{+}^{k}$ 
 and $\bS_{++}^{k}$ denoting the sets of
 real symmetric $k\times k$ positive semi-definite
 and positive definite matrices, respectively. 
 We assume that the observation error covariance matrices 
$\Lambda_i$ are known and that the random effects covariance matrix 
is an unknown diagonal matrix, $\Gamma(\gam) := \Diag{\gamma}, \ \gamma \in \R^s_+$, where, for any vector $\gam\in \Rq$, $\Diag{\gamma}$ is the diagonal matrix with diagonal $\gam$.
 
Define $\omega_i := Z_i u_i + \varepsilon_i$ to be the unknown cluster-specific error vectors. Then model \eqref{eq:lme_setup} can also be viewed as a correlated noise model with
\[ 
	Y_i = X_i\beta + \omega_i,\quad \omega_i \sim \NN(0, \Omega_i(\Gamma)), \quad \Omega_i(\Gamma) = Z_i\Gamma Z_i^T + \Lambda_i.
\] 
This yields the marginalized  negative log-likelihood function of a linear mixed-effects model~(\cite{Patterson1971}):
\eq{
	\label{eq:lmm_objective}
	\LL_{ML}(\beta, \Gamma)  :=
	 \sum_{i = 1}^m \half(y_i - X_i\beta)^T\Omega_i(\Gamma)^{-1}(y_i - X_i\beta) + \half\ln{\det{\Omega_i(\Gamma)}}.
}
Maximum likelihood estimates for $\beta$ and $\Gam$ 
solve the problem
\eq{
	\label{eq:ml_lme}
	\min_{\beta, \Gamma} & \ \LL_{ML}(\beta, \Gamma) 
	\quad \mbox{s.t.} \quad \ \Gamma \in \bS_{+}^{q},
}
and when $\Gam=\Diag{\gam}$ the problem becomes
\eq{
	\label{eq:lme_gam}
	\min_{\beta\in\R^p, \gam\in\R^q_+} & \ \LL(\beta,\gam)
	:=\LL_{ML}(\beta, \Diag{\gamma}) .
	}
In this setting, an
entry $\gamma_j$ takes the  value $0$ when the corresponding coordinates of all random effects $u_{ij}$ are identically $0$ for all $i$, or equivalently,
the randomness in $Y_i$ is completely explained by $\eps_i$. 
The existence of solutions to \eqref{eq:lme_gam}
and, more generally, \eqref{eq:ml_lme} follows  
from the techniques developed in \cite{zheng2021trimmed}.
 

\begin{theorem}[Existence of a Minimizer]\label{thm:basic existence}
Let the assumptions in the statement of problem \eqref{eq:ml_lme} hold. Then optimal solutions to
\eqref{eq:ml_lme} exist.
\end{theorem}

A standalone proof of Theorem \ref{thm:basic existence} that allows us to
give a simple extension to the variable selection case 
is given in Appendix \ref{adx:basic existence}.

We approach feature selection for the model 
\eqref{eq:lme_setup}-\eqref{eq:lme_gam} 
by adding a regularizer to 
the objective yielding an 
optimization problem of the form
\eq{
    \label{eq:vs_lme}
    \min_{x} & \LL\bg + R\bg+ \delta_{\CC}\bg,
}
where 
$\CC:=\R^p\times \R^q_+$,
$\map{R}{\R^P\times\R^q_+}{\eR_+:=\R_+\cup\{+\infty\}}$
is a 
lower semi-continuous (lsc) regularization term, and
$\delta_{\CC}$ is the convex indicator function
\[
    \delta_{\CC}(x) := \begin{cases} 0, &  x \in \CC \\ +\infty, & x \not\in \CC .\end{cases}
\]  
In practice, it is often advisable to include a constraint of the form 
$\gam\le\gammax$ for $\gam\in\R^q_{++}$ chosen
sufficiently large since an excessively large variance usually indicates 
that the model is poorly posed and needs review. 
Such a constraint is also numerically expedient since it prevents $\gam$ from diverging. 
We return to this issue when our algorithm is specified.
We have the following extension to Theorem \ref{thm:basic existence2}
which tells us that solutions to \eqref{eq:vs_lme}
exist whenever $R$ is level compact. 
The proof appears in Appendix \ref{adx:basic existence}.

\begin{theorem}\label{thm:basic existence2}
Let the assumptions in the statement of problem \eqref{eq:lme_gam} hold.
Suppose $\map{\hR}{\R^p\times\R^q_+}{\R\cup\{+\infty\}}$ 
is lsc and level compact (i.e., $\epi{R}:=\lset{(\bg,\nu)}{R\bg\le\nu}$ is closed
and $\lset{\bg}{R\bg\le\nu}$ is bounded for all $\nu\in\R$).
Then $\LL+\hR$ is level compact and solutions to
the following optimization problem exist:
\eq{
\label{eq:extended loss}
\min_{\beta\in\R^p, \gamma\in\R^q_+}  \LL(\beta, \gamma) + 
\hR(\beta, \gamma) .
}
\end{theorem}

\begin{algorithm}[H]
\SetAlgoLined
$x = x_0$, $L>0$, 
$\alpha < \frac{1}{L}$\\ 
 \While{not converged}{
    $x^+ = \prox_{\alpha R + \delta_{\CC}}(x - \alpha \nabla \LL(x))$;\\
 }
 \caption{\label{alg:pgd_for_lme}Proximal Gradient Descent (PGD) for 
 LMEs} 
\end{algorithm}

Since $\LL$ is smooth on its domain, a standard approach to solving
\eqref{eq:vs_lme} is the proximal gradient descent (PGD) algorithm described in 
Algorithm \ref{alg:pgd_for_lme}, where, for $x=\bg$ and $w=\tbg$,
\[\begin{aligned}
\prox_{\alpha R + \delta_{\CC}}(x-\alf\nabla \LL(x))
&:=\argmin_{w\in\CC} \alpha R(w) 
+\frac{1}{2\alf}\Vert w-(x-\alf \nabla f(x))\Vert^2 
\\ 
&=\argmin_{w\in\CC} f(x)+\ip{\nabla f(x)}{w-x}+
R(w)+\frac{1}{2\alf}\Vert w-x\Vert^2.
\end{aligned}
\]
Here the parameter $L$ is intended
to be a global Lipschitz constant for $\nabla \LL$ over its domain. 
Unfortunately, $\nabla \LL$ is not globally Lipschitz on its domain. 
Nonetheless, it is possible to obtain convergence
results with the inclusion of a line search or trust region strategy 
\cite{Burke-Engle18}.
In this paper a different approach is explored that uses
global variational information on $\LL$ rather than
the local linearizations for $\LL$ which form the basis of the PGD algorithm and its
variants.

\section{Relaxation of the mixed-effects variable selection model}
\label{sec:relax}

Our strategy for obtaining approximate solutions to the
mixed-effects variable selection problem  \eqref{eq:vs_lme} is
motivated by the sparse relaxed regularization regression (SR3) strategy
developed in \cite{Zheng2019SR3}. That is, we introduce 
auxiliary variables to decouple two competing goals -- variable selection and data fitting. In addition, we add a barrier term to 
relax the constraint $0\le \gam$.
The
decoupling uses the coupling function 
$\map{\kappa_\eta}{\Rp\times\Rq}{\R}$
given by 
\begin{equation}\label{eq:kappa}
\kappa_\eta(y,z):=\frac{\eta}{2}\norm{(y,z)}^2 ,
\end{equation}
where $\eta\in\R_+$, while the constraint $0\le\gam$ is relaxed
using the perspective of the negative log, i.e 
$\map{\vphi}{\R^q\times\R}{\R\cup\{\infty\}}$ given by
\[
\vphi(\gam,\mu):=
\begin{cases}
-\mu\sum_{i=1}^q\ln (\gam_i/\mu)&,\ \mu>0,\\
\del_{\R^q_+}(\gam)&,\ \mu=0,\\
+\infty&,\ \mu<0.
\end{cases}
\]
The mapping $\vphi$ is known to be a closed proper convex function and, 
for $\mu>0$, it is
essentially equivalent to the well-known log-barrier function. 
For more information on the perspective mapping, its calculus,
and perspective duality, we refer the reader to \cite{ABDFM18,ABF13}.
We call $\eta$ the coupling parameter and $\mu$ the log-barrier parameter
and write $\phimu(\cdot):=\vphi(\cdot,\mu)$.
The relaxed problem employs 
auxiliary variables $\tbg$ 
and relaxation parameters $0\le\eta$ and $0\le \mu$
to obtain the problem 
\begin{equation}\label{eq:relax} 
\begin{aligned}
\min_{\bg, \tbg} & \LL(\beta, \gamma) 
	+\phi_\mu(\gam)+\kappa_\eta(\beta-\tbeta,\gam-\tgam)
+ R(\tbeta, \tgamma) 
\\
\text{s.t. } & \tgamma \geq 0\ .
\end{aligned}
\end{equation}
We rewrite \eqref{eq:relax} 
so as to separate the smooth and nonsmooth components to obtain
\eq{ \label{eq:relax2}
    \min_{\bg, \tbg} & \LL_{\emu}(\bg,\tbg) + R\tbg +\del_{\R^q_+}(\tgam)\ ,
}
where 
\eq{\label{eq:LLlam}
    \LL_{\emu}(\bg, \tbg) := 
    \LL\bg + \phimu(\gam)+\keta(\beta - \tbeta, \gamma - \tgamma).
}
Observe that, for all $\mu,\eta\in\R_+$, $\LL_{\emu},\ \nabla \LL_{\emu}$ and $\nabla^2 \LL_{\emu}$ are continuous
on $(\Rp\times\dom(\phimu))\times(\Rp\times\Rq)$ 
(see Appendix \ref{appendix:derivatives_of_lmm}) 
so that $\LL_{\emu}$ is smooth on its domain.
As in \cite{Zheng2019SR3}, we use the 
decoupling to write \eqref{eq:relax2}
as an iterated optimization problem over the smooth components of the objective.
This yields a representation of the form
\begin{equation}
    \label{eq:relax3}
    \min_{\tbg} \uemu\tbg + R\tbg+\del_{\Rqp}(\tgam),
\end{equation}
where 
\begin{equation}
    \label{eq:value_function_definition}
    \uemu\tbg := \min_{\bg } 
    \LL_{\emu}(\bg, \tbg)\ .
\end{equation}
This is the formulation of the mixed-effects variable selection problem we study.
Our focus is on the optimal value function $\uemu$ which captures global 
variational information about the function $\LL$ over its domain.
We show that $\uemu$ has a locally Lipschitz continuous gradient
and that the evaluation of $\uemu$ and $\nabla \uemu$ is accomplished
by optimizing a well conditioned strongly convex function. This
allows us to apply the PGD algorithm to the function $\uemu$ rather than
the function $\LL$. Our numerical studies show that the global information 
captured by $\uemu$ significantly improves both the accuracy of the solution
obtained and the overall numerical efficiency of the algorithm.

\section{Proximal Gradient Descent for $\uemu\tbg + R\tbg+\del_{\Rqp}(\tgam)$} 
\label{sec:pgd}

We follow the analysis of the PGD algorithm given in \cite[Chapter 10]{AB17}
as it applies to the objective
\begin{equation}\label{eq:varphi}
\Phi_\emu\tbg:=\uemu\tbg +\tR\tbg,\ \ \text{where}\ 
\tR\tbg:=R\tbg+\del_{\Rqp}(\tgam).
\end{equation}
%
Since $u_\emu$ is nonconvex, one typically 
applies a line search method to select stepsize. 
However, this is often not required in practice. 
For this reason we state the algorithm with and without a line search.
\medskip



\begin{algorithm}[H]
\SetAlgoLined
$\tw = \tw^0$, $\gammax>\tgam^0$, $\alpha < \frac{2}{L}$ 
for $L>0$.\\
 \While{not converged and $\tgam\le\gammax$}{
    $\tw^+ = \prox_{\alpha \tR }(\tw - \alpha \nabla \uemu (\tw))$\\
 }
\caption{\label{alg:MSR3}Proximal Gradient Descent for 
$\Phi_\emu$}
\end{algorithm}
\medskip


\begin{algorithm}[H]
\SetAlgoLined
{\bf Initialize:} 
$\theta\in(0,1),\ \tau\in(0,1)$,
$\eta>0,\ \mu>0$, $\eps_{\mbox{\tiny Tol}}\ge 0$
$k=0$, $t_0>0$, $\bw^{0}=(\tbeta^0,\tgam^0)\in \Rp\times\R^q_+$
with $\inf\Phi_\emu <\Phi_\emu(\bw^0)$, 
$\gammax>\tgam^0$, 
$w^0=\prox_{t_0\tR}(\bw^{0}-t_0\nabla u_\emu(\bw^0))$.
\\
\While{{\small $\norm{w^k-\bw^k}>\eps_{\mbox{\tiny Tol}}$
and $\gam^k\le\gammax$}}{\small\smallskip
\begin{enumerate}
\item[(i)] $t_{k+1}=\max\lset{t}{
\begin{aligned}&s\in\bW,\, t=t_0\theta^s,\ 
w=\prox_{t\tR}(w^k-t\nabla u_\emu(w^k))\\ &
\phi(w)
\le \phi(w^k)-\tau t\norm{w^k-w}^2
\end{aligned}}$.
\item[(ii)] $\bw^{k+1}=w^{k}$
\item[]$w^{k+1}=\prox_{t_{k+1}\tR}(w^k-t_{k+1}\nabla u_\emu(w^k))$
\item[]$k=k+1$
\end{enumerate}
}
\smallskip
\caption{\label{alg:pgd with bt}Proximal Gradient Descent 
fo $\Phi_\emu$ with Backtracking}
\end{algorithm}
\medskip

In Algorithm \ref{alg:MSR3}, 
the parameter $L$ is assumed to be a global Lipschitz
constant for $\nabla \uemu$. In Section \ref{sec:convergence pgd}, we show that the existence of $L$ is not needed.
In both algorithms we introduce the requirement that $\gam^k\le\gammax$.
While it is possible to include an explicit constraint of this form in the 
optimal variable selection problem \eqref{eq:vs_lme}, 
we do not do so since we assume that $\gammax$ is chosen 
so large that, from a practical perspective, 
the violation of this constraint indicates that the model is 
poorly posed and the algorithm needs to be terminated.
We base our analysis of the convergence properties of
Algorithms 1 and 2 on \cite[Theorem 10.15]{AB17} which makes use
of the following three basic assumptions:
\smallskip

\noindent
{\bf Basic Assumptions for the PGD Algorithm}
\begin{enumerate}
\item[(A)] $\map{\tR}{\Rp\times\Rq}{\eR}$ is a closed proper convex function.
\item[(B)] $\map{\uemu}{\Rp\times\Rq}{\eR}$ is closed and proper, 
$\dom{\uemu}$ is
convex, $\dom{\tR}\subset\intr{\dom{\uemu}}$, and
$\uemu$ is $L_\emu$-smooth over $\intr{\dom{\uemu}}$.
\item[(C)] Problem \eqref{eq:relax3} has an optimal solution with optimal 
value $\Phi_{\mbox{\tiny OPT}}$.
\end{enumerate}

\noindent
We assume that (A) holds. This is not an overly
restrictive assumption since it is satisfied by most of the standard
variable selection regularizers.
We show that (C) holds when $R$ satisfies
an additional coercivity hypothesis (Theorem \ref{thm:relaxed existence}). 
On the other hand, establishing that (B) holds
in a concrete setting such as ours
can be quite difficult. 
In particular, just as with $\LL$, $\uemu$ may fail to be globally Lipschitz.
Validating Assumption (B) as well as developing a technique for circumventing
the need for a global Lipschitz constant for $\nabla\uemu$ 
consumes the majority of the theoretical development. 


\section{The Smoothness of $u_\emu$}\label{sec:smoothness}
We investigate the relationship between the problems \eqref{eq:vs_lme}
and \eqref{eq:relax2}, the existence of solutions to 
 \eqref{eq:relax2}, and the properties of the function $\uemu$ and its derivative.

\subsection{Underlying convexity}

\begin{lemma}[$\LL+\phimu$ is Weakly Convex]\label{lem:LL weak cvx}
Let $\LL$ be as given in \eqref{eq:lme_gam}. Then
\begin{equation}\label{eq:hess LL}
\nabla^2\LL{(\beta,\gam)}=\sum_{i=1}^m
S_i^T\begin{bmatrix}X_i^T\\ -Z_i^T\end{bmatrix}
\Omega_i(\gam)^{-1}
\begin{bmatrix}X_i& -Z_i\end{bmatrix}S_i
-\begin{bmatrix}
0&0\\ 0& \half(Z_i^T\Omega_i(\gam)^{-1}Z_i)^{\circ2}
\end{bmatrix},
\end{equation}
for all
$(\beta,\gam)\in\R^p\times\R^q_+$, where 
\[
S_i:=\begin{bmatrix}
I_q&0\\ 0&\Diag{Z_i^T\Omega_i^{-1}(X_i\beta-Y_i)}
\end{bmatrix}
\]
and, for any $A\in\R^{t\times t}$, $A^{\circ 2}:=A\circ A$. In particular, this implies
that the matrix
\begin{equation}\label{eq:psd for LL}
\begin{bmatrix}
\nabla_{\beta\beta}\LL(\beta,\gam)&\nabla_{\gam\beta}\LL(\beta,\gam)\\
\nabla_{\beta\gam}\LL(\beta,\gam)&\nabla_{\gam\gam}\LL(\beta,\gam)
+\blam I\end{bmatrix}
\end{equation}
is positive semidefinite
for $\bar\eta= \nu m$, where
\[
\nu:=max\lset{(1/2) \mu_\mmin(\Lam_i)^{-2}\sig^4_\mmax(Z_i)}{i=1,\dots,m},
\]
$\mu_\mmin(\Lam_i)$ is the smallest eigenvalue of $\Lam_i$, and
 $\sig_\mmax(Z_i)$ is the largest singular value of $Z_i,\, i=1,\dots,m$.
Consequently, for any $\tbg\in\dom{\tR}$ and $\mu\ge 0$, the mapping 
\(
\bg\mapsto \LL_{\emu}(\bg, \tbg)
\)
is convex for all $\eta\ge\bar\eta:=\nu m$.
In particular, this implies that $\LL+\phimu$ is weakly convex for any $\mu\ge 0$, 
and the mapping
\(
\bg\mapsto\LL_{\emu}(\bg,\tbg)
\)
is strongly convex for $\eta> \bar\eta$ with modulus of strong convexity $(\eta-\bar\eta)$
regardless of the choice of $\tbg\in\Rp\times\Rq$.
\end{lemma}
\begin{proof}
The formula for $\nabla^2\LL$ is given in Appendix \ref{appendix:derivatives_of_lmm}
(see \eqref{eq:hess LL}).
By \cite[Theorem 3.1]{ABBP2021}, 
$\mu_\mmax(\left(Z_i^T\Omega_i(\gam)^{-1}Z\right)\le
\lam_\mmin^{-1}\sig_\mmax^2(Z_i)$, and since 
$\mu_\mmax(H^{\circ 2})\le \mu_\mmax^2(H)$
for all $H\in\bS^q_+$ \cite{HJ85}, we have
\[
\mu_\mmax\left(\half(Z_i^T\Omega_i(\gam)^{-1}Z_i)^{\circ2}\right)
\le (1/2) \lam_\mmin^{-2}
\sig^4_\mmax(Z_i)=:\nu_i
\qquad  i=1,\dots,m.
\]
This establishes that the matrix in \eqref{eq:psd for LL} is positive semidefinite.
Since $\phimu$ is convex, the mapping
\[
(\beta,\gam)\mapsto \LL_{\emu}(\bg,\tbg) 
\]
is strongly convex for any choice
of $\blam> \nu m$, where 
\eq{\label{eq:nu}
\nu:=\max_{i=1,\dots,m}\nu_i.
}
\end{proof}
For the remainder of the paper, we assume that 
\eq{\label{eq:blam1}\eta> \nu m=:\bar\eta} 
so that the mapping
$(\beta,\gam)\mapsto \LL_{\emu}(\bg,\tbg)$ 
is strongly convex with 
positive definite Hessian
regardless of the 
choice of $\tgam\in\R^q$. With this in mind, the function $\uemu$ defined by
\eqref{eq:value_function_definition} resembles a Moreau envelope.
However, this is misleading since, in particular,
we are not even assured of the existence of solutions to
the optimization problem defining $\uemu$. 

\subsection{Existence and consistency}\label{ssec:existence and consistency}

To establish the existence of solutions to the relaxed
optimization problems \eqref{eq:relax2} and the
problems defining the parametrized family $\uemu$ in
\eqref{eq:value_function_definition}, we assume that $R$ is $1$-coercive.
\begin{lemma}\label{lem:lb plus 1c}
Given $\mu>0$ let $\phimu$ be as defined above, and assume that 
$\map{R}{\Rp\times\Rq}{\R\cup\{+\infty\}}$
is 1-coercive, i.e., 
\(
\liminf_{\Vert{\tbg}\Vert\rightarrow\infty}\Vert{\tbg}\Vert^{-1}R\tbg>0.
\)
Then 
$\phimu+R$ 
is level compact.
\end{lemma}
%
\begin{proof}
If $\mu=0$, then the result is trivially true, so we assume that $\mu>0$.
Let $\{\bgk\}\subset \Rp\times\Rqp$ be such that $\norm{\bgk}\uparrow\infty$.
We need to show that $\phimu(\gam^k)+R\bgk\rightarrow \infty$.
If $\{\gam^k\}$ is bounded, then $\phimu(\gam^k)+R\bgk\rightarrow \infty$
since in this case $\phimu(\gam^k)$ is bounded below.
So assume that $\{\gam^k\}$ is unbounded which implies that
$\phimu(\gam^k)\rightarrow -\infty$. Since $R$ is 1-coercive,
we know that there is an $\hat\alf>0$ such that, for $k$ sufficiently large,
$R\bgk\ge\hat\alf\sum_{i=1}^q\gam^k_i$. But then
$\phimu(\gam^k)+R\bgk\ge \sum_{i=1}^q(\hat\alf\gam^k_i-\mu\ln(\gam^k_i))$
where the right-hand side diverges to $+\infty$ as $k\uparrow\infty$.
Hence, $\phimu(\gam^k)+R\bgk\rightarrow \infty$.
\end{proof}


\begin{theorem}\label{thm:relaxed existence}
Let $\LL$  be as in Theorem \ref{thm:basic existence2}
and let $\eta> 0$ satisfy \eqref{eq:blam1}.
Let $\mu\ge 0$. If $\mu=0$, assume that $\map{R}{\Rp\times\Rq}{\R_+\cup\{+\infty\}}$
is level compact; otherwise, assume $R$
is 1-coercive.
Then solutions to \eqref{eq:relax2} always exist.
\end{theorem}
\begin{proof}
Let 
\(
v^* 
\)
be the optimal value in \eqref{eq:relax2}
and let
\(
\!\{(\bgk,\tbgk)\}\!\subset (\Rp\times\Rqp)^2
\) 
be such that
\(
\LL_\emu(\bgk,\tbgk)\!+\!R\tbgk\downarrow v^*\!.
\)
By \eqref{eq:eigbd} and \eqref{eq:eig1}, 
it must be the case that 
\begin{equation}\label{eq:lower bd}
\begin{aligned}
\LL_\emu&(\bgk, \tbgk)+R\tbgk
\\ &\ge
\frac{n\!+\!1}{2}\ln(\talf)+\phimu(\gam^k)+
\keta(\beta^k-\tbeta^k,\gam^k-\tgam^k)+R\tbgk
\\ &\ge
\frac{n\!+\!1}{2}\ln(\talf)+\phimu(\gam^k)+\frac{\bar\eta}{2}\norm{\gam^k-\tgam^k}^2
+R\tbgk.
\end{aligned}
\end{equation}

If $v^*=-\infty$, then \eqref{eq:lower bd} tells us that
\begin{equation}\label{eq:minus infty1}
\phimu(\gam^k)
+\frac{\blam}{2}\norm{\gam^k-\tgam^k}^2+R\tbgk\rightarrow-\infty.
\end{equation}
This in turn implies that $\mu>0$, $\phimu(\gam^k)\rightarrow-\infty$
and $\norm{\gam^k}\rightarrow \infty$.
Since $R$ is 1-coercive and $\norm{\gam^k}\rightarrow \infty$, 
we can assume with no loss in generality that there is an $\balf>0$
such that $R\tbgk\ge \balf \sum_{i=1}^q\tgam^k_i$ for all $k\in\N$. 
Consequently,
\eq{\label{eq:unbdd1}\begin{aligned}
\phimu(\gam^k)
+&\frac{\blam}{2}\norm{\gam^k-\tgam^k}^2+R\tbgk
\\ &\ge
\sum_{i=1}^q\left(-\mu\ln (\gam^k_i/\mu)+\frac{\blam}{2}(\gam^k_i-\tgam^k_i)^2
+\balf\tgam^k_i\right)
\\ &= 
\sum_{i=1}^q\left(
(-\mu\ln (\gam^k_i/\mu)+\balf\gam^k_i)+\frac{\blam}{2}(\gam^k_i-\tgam^k_i)^2
-\balf (\gam^k_i-\tgam^k_i)
\right)
\\ &=
\sum_{i=1}^q\left(
(-\mu\ln (\gam^k_i/\mu)+\balf\gam^k_i)+\frac{\blam}{2}
\left[(\gam^k_i-\tgam^k_i-\frac{\balf}{\blam})^2-(\frac{\balf}{\blam})^2\right]
\right)
\\ &\ge -q\frac{\balf^2}{2\blam}+\sum_{i=1}^q
(-\mu\ln(\gam^k_i/\mu)+\balf\gam^k_i)
\\ &= -q\frac{\balf^2}{2\blam}+\phimu(\gam^k)+\balf\norm{\gam^k}_1\ 
\rightarrow +\infty,
\end{aligned}}
which is a contradiction. Hence $v^*>-\infty$.

Let $\rho> v^*>-\infty$. If $\{\gam^k\}\subset\Rqp$ is unbounded, we may assume with no
loss in generality that $\norm{\gam^k}\rightarrow+\infty$. 
If $\mu=0$, then, by  \eqref{eq:lower bd},
\(
\rho>\frac{n\!+\!1}{2}\ln(\talf)+R\tbgk
\uparrow +\infty,
\)
a contradiction, and so we can assume that $\mu>0$ and $R$ is 1-coercive.
Using \eqref{eq:lower bd} we may proceed as in \eqref{eq:unbdd1} to find that
\begin{equation}\label{eq:lbd2}
\rho > \frac{n\!+\!1}{2}\ln(\talf)
-q\frac{\balf^2}{2\blam}+\sum_{i=1}^q
(-\mu\ln \gam^k_i+\balf\gam^k_i)\ 
\rightarrow +\infty,
\end{equation}
again a contradiction, so the sequence $\{\gam^k\}$ is bounded.
Therefore, the first inequality in \eqref{eq:lower bd} tells us that the entire sequence 
$\{(\bgk,\tbgk)\}$ is necessarily bounded.
Consequently, a limit point of the sequence $\{(\bgk,\tbgk)\}$ exists and,
since $R$ is lsc, any such
limit point is a solution to 
\eqref{eq:relax2}.
\end{proof}

Next we fix $\mu\ge 0$ and show that as 
$\eta\uparrow\infty$ the solutions to 
\eqref{eq:relax2} converge to solutions of
\begin{equation}\label{eq:log-barrier problem}
\min_{\bg\in\CC}\LL\bg+\phimu(\gam)+R\bg.
\end{equation}
In particular, for $\mu=0$, they converge to solutions of 
\eqref{eq:vs_lme}.

\begin{theorem}[Consistency as $\eta\rightarrow\infty$]\label{thm:eta consistency}
Let $\LL$ and $R$ be as in Theorem \ref{thm:relaxed existence}
and fix $\mu\ge 0$.
Let $\{\eta_k\}\subset\R_{++}$ be such that $\eta_k<\eta_{k+1}$
with $\eta_k\uparrow\infty$, and let $(\bgk,\tbgk)$ be an optimal solution to 
\eqref{eq:relax2} for $(\emu)=(\eta_k,\mu)$, $k\in\N$.
Then any limit point (equivalently, cluster point) 
$(\bbg,\hbg)$ of $\{(\bgk,\tbgk)\}$ satisfies
$\bbg=\hbg$ with $\bbg$ being an optimal solution to 
\eqref{eq:log-barrier problem}.
\end{theorem}
\begin{proof}
With no loss in generality $\eta_k>\bar\eta$ for all $k$.
Set 
\[\left.\begin{aligned}
a_k(x,w)&:=\LL_{\eta_k,\mu}(x,w)+R(w)\\ 
b_k(x,w)&:=\LL(x)+\phimu(\gam)+R(w)\\
c_k(x,w)&:=\kappa_{\eta_k}(\beta-\tbeta,\gam-\tgam)
\end{aligned}\right\}\quad \forall k\in\N,
\] 
where 
$x=(\beta,\gam)$ and $w=(\tbeta,\tgam)$ with $\kappa_\eta$ defined in
\eqref{eq:kappa}. Set $x^k=\bgk,\ \bx=\bbg,\ w^k=\tbgk$ and $\hw=\hbg$.
By Lemma \ref{lem:lb plus 1c} and Theorem \ref{thm:basic existence2} 
with $\hR=\phimu+R$, 
there is an optimal solution $\xmu$ 
to \eqref{eq:log-barrier problem} 
yielding an optimal value of $\vmu$ for which 
$a_k(x^k,w^k)\le a_k(\xmu,\xmu)=\vmu$ for all $k\in\N$. Hence, the sequence
$\{a_k(x^k,w^k)\}$ is upper bounded by $\vmu$.
Since
\[
a_k(x^k,w^k)\le a_k(x^{k+1},w^{k+1})\le a_{k+1}(x^{k+1},w^{k+1}),
\]
there exists $\tv$ such that
$a_k(x^k,w^k)\uparrow\tv\le \vmu$.
Next, observe that
\[
a_k(x^k,w^k) \le a_k(x^{k+1},w^{k+1})\ \ \text{ and }\ \ 
a_{k+1}(x^{k+1},w^{k+1})\le a_{k+1}(x^{k},w^{k}).
\]
By adding these inequalities together we find that 
$\norm{x^{k+1}-w^{k+1}}\le \norm{x^{k}-w^{k}}$ so that 
$\norm{x^{k}-w^{k}}\downarrow\tkappa$ for some $\tkappa\ge 0$.
We also have
\[\begin{aligned}
b_k(x^k,w^k)+(\eta_k/2)\norm{x^{k}-w^{k}}
&=a_k(x^k,w^k)
\\ &\le a_k(x^{k+1},w^{k+1})
\\ &=b_{k+1}(x^{k+1},w^{k+1})+(\eta_k/2)\norm{x^{k+1}-w^{k+1}}
\\ &\le b_{k+1}(x^{k+1},w^{k+1})+(\eta_k/2)\norm{x^{k}-w^{k}},
\end{aligned}\]
which gives $b_k(x^k,w^k)\le b_{k+1}(x^{k+1},w^{k+1})\le \tv$.
Therefore, $b_k(x^k,w^k)\uparrow \hv$ for some $\hv\le\tv$.
Consequently,
\[
\tkappa=\lim_k\norm{x^{k}-w^{k}}=
\lim_k\eta_k^{-1}[a_k(x^k,w^k)-b_k(x^k,w^k)]=0.
\]
Therefore, if $(\bx,\bw)$ is any limit point of the sequence $\{(x^k,w^k)\}$,
then $\bx=\bw$ and $\LL(\bx)+\phimu(\bgam)+R(\bx)=\vmu$ 
since $\LL(x^k)+\phimu(\gam^k)+R(w^k)\le a_k(x^k,w^k)\le \vmu$
for all $k\in\N$.
\end{proof}
We now pair Theorem \ref{thm:eta consistency} with a consistency result for the barrier parameter $\mu$.

\begin{theorem}[Consistency as $\mu\rightarrow 0$]
\label{thm:mu consistency}
Let $\LL$ and $R$ be as in Theorem \ref{thm:relaxed existence}. 
For every $\mu\ge 0$, problem 
\eqref{eq:log-barrier problem} has a solution $(\beta_\mu,\gam_\mu)$.
Moreover, if $\{\mu_k\}\subset\R_{++}$ is such that $\mu_k\downarrow 0$,
then the sequence $\{(\beta_{\mu_k},\gam_{\mu_k})\}$ is bounded
and every limit point of the sequence 
is a 
solution to \eqref{eq:vs_lme}.
\end{theorem}
\begin{proof}
The existence of $(\beta_\mu,\gam_\mu)$ for all $\mu\ge 0$ 
follows immediately from 
Lemma \ref{lem:lb plus 1c} and Theorem \ref{thm:basic existence2} with
$\hR=R+\phimu$. 
Let $\mu_k\downarrow 0$
and set $(\beta^k,\gam^k):=(\beta_{\mu_k},\gam_{\mu_k})$.
Set $\tLL:=\LL+R+\del_{\Rp\times\Rqp}$ so that 
the objective in \eqref{eq:log-barrier problem} is $\tLL+\phimu$
and the objective in \eqref{eq:vs_lme} is $\tLL$
with $(\beta_0,\gam_0)$ a solution to \eqref{eq:vs_lme}
by definition. 
Observe that
\[\begin{aligned}
\tLL\bgk+\phi_{\mu_k}(\gam^k)&\le 
\tLL(\beta^{k+1},\gam^{k+1})+\phi_{\mu_k}(\gam^{k+1})\quad \text{and}
\\
\tLL(\beta^{k+1},\gam^{k+1})+\phi_{\mu_{k+1}}(\gam^{k+1})&\le 
\tLL(\beta^{k},\gam^{k})+\phi_{\mu_{k+1}}(\gam^k)
\end{aligned}\]
Summing these inequalities yields the inquality
\[
(\mu_k-\mu_{k+1})\sum_{i=1}^q\ln(\gam^{k}_i)
\ge 
(\mu_k-\mu_{k+1})\sum_{i=1}^q\ln(\gam^{k+1}_i),
\]
so $\{\sum_{i=1}^q\ln(\gam^{k}_i)\}$ is a non-increasing sequence.
Therefore,
\[\begin{aligned}
\tLL(\beta^{k+1},\gam^{k+1})+\phi_{\mu_{k+1}}(\gam^{k+1})&\le 
\tLL(\beta^{k},\gam^{k})+\phi_{\mu_{k+1}}(\gam^k)
\\ &\le
\tLL(\beta^{k},\gam^{k})+\phi_{\mu_{k+1}}(\gam^{k+1})
\end{aligned}\]
which implies that $\{\tLL(\beta^{k},\gam^{k})\}$
is also a non-increasing sequence and bounded below by 
$\tLL(\beta_0,\gam_0)$. 
Since
Theorem \ref{thm:basic existence2} tells us that $\tLL$ is level compact,
the sequence $\{\bgk\}$ is bounded. 
Let $(\bbeta,\bgam)\in\Rp\times\Rqp$ be any
limit point of $\{\bgk\}$ and let $J\subset\N$ be such that
$\bgk\overset{J}{\rightarrow}(\bbeta,\bgam)$. Then
\[
\tLL\bgk+\phi_{\mu_k}(\gam^k)\le
\tLL\bg+\phi_{\mu_k}(\gam) \quad \forall\,\bg\in\Rp\times\R^q_{++}.
\]
Since $\tLL$ is continuous on $\Rp\times\Rqp$ and the perspective function
$\phi_{\mu}(\gam)=\varphi(\mu,\gam)$ is lsc on $\Rp\times\Rqp$, we have
\[
\tLL(\bbeta,\bgam)\le\liminf_{k\in J}(\tLL\bgk+\phi_{\mu_k}(\gam^k))\le
\tLL\bg \quad \forall\,\bg\in\Rp\times\R^q_{++}.
\]
Consequently, the continuity of $\tLL$ on $\Rp\times\Rqp$ implies
that $(\bbeta,\bgam)$ solves \eqref{eq:vs_lme}.
\end{proof}

\subsection{The continuity and differentiability of $\uemu$}

The continuity of $\uemu$ is closely tied to the continuity of 
the associated solution mapping
$\map{\SS_\emu}{\Rp\times\Rq}{\Rp\times\dom(\phimu)}$ given by 
\begin{equation}\label{eq:argmin for u}
\SS_\emu(\tbeta,\tgam):= 
\argmin_{(\beta,\gam) } \LL_{\emu}(\bg,\tbg)\ .
\end{equation}

\begin{theorem}[Continuity of $\uemu$ and $\SS_\emu$]\label{thm:u cont}
Let the assumptions of Theorem \ref{thm:relaxed existence} hold.
For every $(\mu,\eta)\in\R_+\times\R_{++}$,
the function $\uemu$ defined in \eqref{eq:value_function_definition}
is well-defined and continuous on $\Rp\times\Rq$.
In addition, the solution mapping $\SS_\emu$
is well-defined, single-valued and continuous on $\Rp\times\Rq$.
\end{theorem}

\begin{proof}
Since $\eta>\blam= \nu m$, Lemma \ref{lem:LL weak cvx} 
tells us that the objective in \eqref{eq:value_function_definition}
is strongly convex, 
and so \eqref{eq:value_function_definition} has a unique solution. 
Consequently, $\SS_\emu$ is well-defined and single-valued
on $\Rp\times\Rq$.
This implies that $\uemu$ is also well defined on $\Rp\times\Rq$ since
\[
\uemu(\tbeta,\tgam)= \LL_{\emu}(\SS_\emu(\tbeta,\tgam),\tbg)
    \quad\forall\, (\tbeta,\tgam)\in \Rp\times\Rq.
\]
The result follows once it is shown that $\SS_\emu$ is continuous.

Let $\{(\tbeta^k,\tgam^k)\}\subset\Rp\times\Rq$ and 
$(\tbeta^*,\tgam^*)\in\Rp\times\Rq$ be such that 
$(\tbeta^k,\tgam^k)\rightarrow(\tbeta^*,\tgam^*)$.  
Set $(\hbeta^k,\hgam^k):=\SS_\emu(\tbeta^k,\tgam^k),\ k\in\N$
and $(\bbeta,\bgam)=\SS_\emu(\tbeta^*,\tgam^*)$. 
We must show that $(\hbeta^k,\hgam^k)\rightarrow (\bbeta,\bgam)$.
We begin by showing that the sequence $\{(\hbeta^k,\hgam^k)\}$
is bounded.
By Lemma \ref{lem:LL weak cvx}, the mapping $\bg\mapsto \LL_{\emu}(\bg,\tbg)$ is strongly convex with modulus of strong convexity $\eta$ for all $\tbg\in\Rp\times\Rq$.
In particular, this implies that
\eq{\label{eq:bd limit}\begin{aligned}
\LL_{\emu}((\bbeta,\bgam),\tbgk)&+
\ip{\nabla_{\bg} \LL_{\emu}((\bbeta,\bgam),\tbgk)}
{(\hbeta^k,\hgam^k)-(\bbeta,\bgam)}
\\ &+
\frac{\eta}{2}\norm{(\hbeta^k,\hgam^k)-(\bbeta,\bgam)}^2
\\ &\le
\LL_{\emu}((\hbeta^k,\hgam^k),\tbgk)
\\ &\le
\LL_{\emu}((\bbeta,\bgam),\tbgk).
\end{aligned}}
Since $(\tbeta^k,\tgam^k)\rightarrow(\tbeta^*,\tgam^*)$ and both
$\nabla_{\bg} \LL_{\emu}((\bbeta,\bgam),\cdot)$ and 
$\LL_{\emu}((\bbeta,\bgam),\cdot)$ are continuous at 
$(\tbeta^*,\tgam^*)$, we can assume with no loss in generality that there 
is a constant 
$c>0$ such that
\[
\norm{\nabla_{\bg} \LL_{\emu}((\bbeta,\bgam),\tbgk)}\le c\text{ and }
|\LL_{\emu}((\bbeta,\bgam),\tbgk)|\le c\quad\forall\, k\in \N.
\]
Plugging this into \eqref{eq:bd limit} and simplifying gives
\[
\frac{\eta}{2}\norm{(\hbeta^k,\hgam^k)-(\bbeta,\bgam)}^2
\le c(1+\norm{(\hbeta^k,\hgam^k)-(\bbeta,\bgam)}).
\]
Therefore the sequence $\{(\hbeta^k,\hgam^k)\}$ must be bounded.

Let $(\beta_0,\gam_0)$ be any limit point of $\{(\hbeta^k,\hgam^k)\}$
and let $J\subset\N$ 
be such that
$(\hbeta^k,\hgam^k)\overset{J}{\rightarrow}(\beta_0,\gam_0)$. Then, by
the final inequality in 
\eqref{eq:bd limit}, we can take the limit in $k\in J$ to find
that 
\(
\LL_{\emu}((\beta_0,\gam_0),(\tbeta^*,\tgam^*))
\le
\LL_{\emu}((\bbeta,\bgam),(\tbeta^*,\tgam^*)).
\)
The uniqueness of $(\bbeta,\bgam)$ tells us that
$(\beta_0,\gam_0)=(\bbeta,\bgam)$. 
Since 
$(\beta_0,\gam_0)$ was any limit point of the bounded sequence 
$\{(\hbeta^k,\hgam^k)\}$,
we have $(\hbeta^k,\hgam^k)\rightarrow (\bbeta,\bgam)$ which
implies that $\SS_\emu$ is continuous on $\Rp\times\Rq$.
\end{proof}

We now consider the differentiability of $\uemu$. For this we make use of the
following lemma.

\begin{lemma}[Local uniform level boundedness of $\LL_\emu$]
\label{lem:level bdd}
Let $\mu\ge 0$, $\eta>\blam$ and suppose that the assumptions of Theorem \ref{thm:relaxed existence} hold.
Set $x=\bg$ and $w=\tbg$.
Then
the function
$\LL_\emu(\bg,\tbg)$ is level bounded in $\bg$ locally uniformly in $\tbg$
for all $\tbg\in \Rp\times\Rq$. That is, for every 
$\tbg\in\Rp\times\Rq$ and $\rho\in\R$,  
there are $\nu\in\R$ and $\eps>0$ such that
\(
\lset{\bg}{\LL_\emu(\bg,(\bbeta,\bgam)\le\rho}\subset \nu\B
\)
for all $\tbg\in(\bbeta,\bgam)+\eps\B$.
\end{lemma}
\begin{proof}
Set $x=\bg,\ w=\tbg$, and $\bw=(\bbeta,\bgam)$.
If the result is false, there exists $\bw\in\Rp\times\Rq$, $\rho>0$, and a sequence
$\{(x^k,w^k)\}\subset(\Rp\times\dom(\phimu))\times(\Rp\times\Rq)$ 
such that
$w^k\rightarrow\bw$ and $\norm{x^k}\uparrow\infty$ with 
$x^k\in\lset{x}{\LL_\emu(x,w^k)\le\rho}$
for all $k\in\N$. 
By Lemma \ref{lem:LL weak cvx}, the mappings 
$x\mapsto\LL_\emu(x,w)$ are strongly convex with modulus 
$\hat\eta:=\eta-\blam >0$
for all $w\in\Rp\times\Rq$. Let $\hx\in\Rp\times\dom{\phi_\mu}$. 
Then $(\hx,\bw)$
is a point of continuity for $\nabla_x\LL_\emu$, so with no loss in 
generality there is a $c_0>0$ such that 
$\norm{\nabla_x\LL_\emu(\hx,w^k)}\le c_0$ for all $k\in\N$.
Then strong convexity implies that
\[\begin{aligned}
\LL_\emu(\hx,w^k)&-c_0\norm{x^k-\hx}+\frac{\hat\eta}{2}\norm{x^k-\hx}^2
\\ &\le
\LL_\emu(\hx,w^k)+\ip{\nabla_x\LL_\emu(\hx,w^k)}{x^k-\hx}
+\frac{\hat\eta}{2}\norm{x^k-\hx}^2
\\ &\le
\LL_\emu(x^k,w^k)\le\rho.
\end{aligned}\]
But $\LL_\emu(\hx,w^k)-c_0\norm{x^k-\hx}+\frac{\hat\eta}{2}\norm{x^k-\hx}^2\uparrow\infty$ since $\norm{x^k}\uparrow\infty$. This contradiction 
establishes the result.
\end{proof}

\begin{theorem}[Differentiability of $\uemu$]\label{thm:u diff}
Let $\mu\ge 0$, $\eta>\blam$ and suppose that the assumptions of Theorem \ref{thm:relaxed existence} hold.
Then
the function $\uemu$ defined in \eqref{eq:value_function_definition}
is continuously differentiable on $\Rp\times\Rq$ with 
\begin{equation}\label{eq:grad u}
\nabla \uemu\tbg\!=\!\nabla_{\tbg}\kappa_\eta(\tbeta\!-\!\hbeta,\tgam\!-\!\hgam)
\!=\!\eta
\begin{pmatrix}\tbeta\!-\!\hbeta\\ \tgam\!-\!\hgam\end{pmatrix}
\!,\,
\text{where}\ (\hbeta,\hgam)\!=\!\SS_\emu\tbg.
\end{equation}
\end{theorem}

\begin{proof} 
We show that the result follows from \cite[Theorem 10.58]{rockafellar2009variational}.
Set $x=\bg$ and $w=\tbg$. The objective function in the definition of
$\uemu$ is $\LL_\emu(x,w)$, where $\LL_\emu$  is proper and lsc.
Moreover, Lemma \ref{lem:level bdd} tells us that
$\LL_\emu(x,w)$ is level bounded in $x$ locally uniformly in $w$
for all $w\in \Rp\times\Rq$. We have already observed that, 
for all $\mu\in\R_+$ and $\eta>\blam$, $\LL_{\emu}$ and  $\nabla \LL_{\emu}$ 
are continuous
on $(\Rp\times\dom(\phimu))\times(\Rp\times\Rq)$.
Therefore, by \cite[Theorem 10.58]{rockafellar2009variational}
and Theorem \ref{thm:u cont}, $\uemu$ is locally upper-$\CC^1$
and strictly differentiable at every point $w\in\Rp\times\Rq$ with
$\nabla \uemu(w)=\nabla _w\LL_{\emu}(\SS_\emu(w))$.
In addition, $\SS_\emu$ is continuous on 
$\Rp\times\Rq$.
The result follows since $\nabla _w\LL_{\emu}(\bg,\tbg)
=\eta\begin{pmatrix}\tbeta-\beta\\ \tgam-\gam\end{pmatrix}
$.
\end{proof}

\subsection{The Lipschitz Continuity of $\nabla \uemu\tbg$}

Since our goal is to employ the PGD algorithm to solve the relaxed problems
\eqref{eq:relax}, we 
require that $\nabla \uemu\tbg$ be Lipschitz continuous. 
Formula \eqref{eq:grad u} tells us that the Lipschitz continuity of
$\nabla \uemu\tbg$ is equivalent to that of the solution mapping 
$\SS_\emu$. 
To study the Lipschitz continuity of $\SS_\emu$ we make use of 
the mapping
$\map{G}{\Rp\times\Rq\times\Rp}{\Rp\times\Rq}$ 
be given by
\begin{equation}\label{eq:G}
G_\emu((\beta, \gam,v),(\tbeta,\tgam)):=
\begin{bmatrix}
\nabla_\beta \LL(\beta, \gam) + \eta(\beta - \tbeta) \\
\nabla_\gam \LL(\beta, \gam) + 
\eta(\gam - \tgam) - v
\\
v\odot \gamma  -\mu\one
\end{bmatrix}.
\end{equation}
Observe that, for $\mu>0$,  
$(\hbeta, \hgam)=\SS_\emu(\tbeta,\tgam)$ 
if and only if 
\begin{equation}\label{eq:kkt}
\hgam,\hv\in\R_+^q\ \text{ and }\ G_\emu((\hbeta, \hgam,\hv),(\tbeta,\tgam))=0,
\end{equation} 
since the equation $v\odot \gamma  =\mu\one$ implies that
$v=-\nabla \phimu(\gam)$. In addition, when $\mu=0$, condition \eqref{eq:kkt}
is equivalent to $(\hbeta, \hgam,v)$ being a KKT point for 
the optimization problem in
\eqref{eq:value_function_definition}
which, in turn, is equivalent to $(\hbeta, \hgam)=\SS_\emu(\tbeta,\tgam)$
by Theorem \ref{thm:u cont}. We record these observations in the following lemma.

\begin{lemma}\label{lem:oc for u}
Let the assumptions of Theorem \ref{thm:relaxed existence} hold.
Then, for every $(\mu,\eta)\in\R_+\times\R_{++}$, 
$(\hbeta, \hgam)=\SS_\emu(\tbeta,\tgam)$ if and only if there is a vector
$\hv\in\R^q_+$ such that $G_\emu((\hbeta, \hgam,\hv),(\tbeta,\tgam))=0$.
If $\mu>0$, then $\hv=-\nabla \phimu(\hgam)$, and if $\mu=0$, then
$\hv$ is the unique KKT multiplier associated with the constraint 
$0\le\gam$.
\end{lemma}

Our approach to establishing the Lipschitz continuity of $\SS_\emu$ is to 
first show that $\SS_\emu$ is differentiable and then obtain a bound on its Jacobian.
As usual, diffentiability follows by applying the implicit function theorem to 
$G_\emu$.

\begin{lemma}[The invertibility of $\nabla_{(\beta, \gam,v)}G_\emu$]\label{lem:invert}
Let the assumptions of Theorem \ref{thm:relaxed existence} hold and let
$G_\emu$ be as given in \eqref{eq:G}.
Let 
$(\tbeta,\tgam)\in\Rp\times\Rq$ and  
$(\hbeta, \hgam,\hv)\in\Rp\times\R_+^q\times\R_+^q$ be such that
$G_\emu((\hbeta, \hgam,\hv),(\tbeta,\tgam))=0$. 
Then $\nabla_{(\beta, \gam,v)}G_\emu((\hbeta, \hgam,\hv),(\tbeta,\tgam))$
is invertible if and only if
\begin{equation}\label{eq:str cs}
0<\hv_i+\hgam_i,\ i=1,\dots,q,\qquad\mbox{(strict complementary slackness)}
\end{equation}
which automatically holds if $\mu>0$.
In this case, the inverse is given by
\begin{equation}\label{eq:grad G inv}
\begin{bmatrix}
H^{-1}\!\!-\!\begin{bmatrix}\hat R\\ \hat H_2\end{bmatrix}\!
(D(\hgam)\!+\!D(\hv)\hat H_2)^{-1}D(\hv)[\hat R^T\ \hat H_2]
& 
\!\begin{bmatrix}\hat R\\ \hat H_2\end{bmatrix}\!(D(\hgam)\!+\!D(\hv)\hat H_2)^{-1}
\\
\!-\,(D(\hgam)\!+\!D(\hv)\hat H_2)^{-1}D(\hv)[\hat R^T\ \hat H_2]
&\!
(D(\hgam)\!+\!D(\hv)\hat H_2)^{-1}
\end{bmatrix},
\end{equation}
where $D(\hgam):=\Diag{\hgam}$, $D(\hv):=\Diag{\hv}$, 
\[
\begin{aligned}
H&=
\begin{bmatrix}
H_1&R\\ R^T& H_2
\end{bmatrix}
:=\begin{bmatrix}
\nabla_{\beta\beta}\LL(\hbeta,\hgam)\!+\!\eta I&\nabla_{\gam\beta}\LL(\hbeta,\hgam)
\\
\nabla_{\beta\gam}\LL(\hbeta,\hgam)&\nabla_{\gam\gam}\LL(\hbeta,\hgam)
\!+\!\eta I
\end{bmatrix}\qquad\text{and}
\\
H^{-1}&=
\begin{bmatrix}
H_1^{-1}+H_1^{-1}R(H_2-R^TH_1^{-1}R)^{-1}R^TH_1^{-1}& -H_1^{-1}R(H_2-R^TH_1^{-1}R)^{-1}
\\
-(H_2-R^TH_1^{-1}R)^{-1}R^TH_1^{-1}&(H_2-R^TH_1^{-1}R)^{-1}
\end{bmatrix}
\\
&=:
\begin{bmatrix}
\hat H_1&\hat R\\ \hat R^T& \hat H_2
\end{bmatrix}.
\end{aligned}
\]
\end{lemma}
\begin{proof}
Observe that
\begin{equation}\label{eq:nabla G}
\nabla_{(\beta, \gam,v)}G_\emu((\beta, \gam,v),(\tbeta,\tgam))
\!=\!\!
\begin{bmatrix}
\nabla_{\beta\beta}\LL(\beta,\gam)\!+\!\eta I&
\nabla_{\gam\beta}\LL(\beta,\gam)&0\\
\nabla_{\beta\gam}\LL(\beta,\gam)&\nabla_{\gam\gam}\LL(\beta,\gam)
\!+\!\eta I&-I
\\
0&\Diag{v}&\Diag{\gam}
\end{bmatrix}.
\end{equation}
Let us first assume that 
$\nabla_{(\beta, \gam,v)}G_\emu((\hbeta, \hgam,\hv),(\tbeta,\tgam))$
is invertible and, for simplicity write  
\[
\nabla_{(\beta, \gam,v)}G_\emu((\hbeta, \hgam,\hv),(\tbeta,\tgam))
=\begin{bmatrix}H&A\\ B^T&D\end{bmatrix},
\]
where $A:=[0,\ -I]^T$, $B=[0,\ \Diag{\hv}]^T$ and $D:=\Diag{\hgam}$. 
Since $H\in\bS^{p+q}_{++}$,
the matrix
\begin{equation}\label{eq:grad G inv1}
\begin{aligned}
\begin{bmatrix}I&0\\ -B^TH^{-1}&I\end{bmatrix}
\begin{bmatrix}H&A\\ B^T&D\end{bmatrix}
\begin{bmatrix}I&-H^{-1}A\\ 0&I\end{bmatrix}&=
\begin{bmatrix}H&0\\ 0&D-B^TH^{-1}A\end{bmatrix}
\\ &=\begin{bmatrix}H&0\\ 0&\Diag{\hgam}+\Diag{\hv}\hat H_{2}\end{bmatrix}
\end{aligned}
\end{equation}
is nonsingular. In particular, the matrix $\Diag{\hgam}+\Diag{\hv}\hat H_{2}$ is
necessarily invertible.
But if there is an $i$ such that $0=\hgam_i+\hv_i$, then $\hgam_i=\hv_i=0$
so that the matrix $\Diag{\hgam}+\Diag{\hv}H_{22}$ has a zero row and so
is singular. Since this cannot be that case,
\eqref{eq:str cs} must hold.

Conversely, suppose $(r^T,s^T,t^T)^T$ is in the nullspace of 
$\nabla_{(\beta, \gam,v)}G_\emu((\hbeta, \hgam,\hv),(\tbeta,\tgam))$. Then
$0=\Diag{\hv}s+\Diag{\hgam}t$. This combined with \eqref{eq:str cs} implies
that $s^Tt=0$. Consequently,
\[
0=\begin{pmatrix}r\\ s\end{pmatrix}^T\begin{pmatrix}0\\ t\end{pmatrix}
=\begin{pmatrix}r\\ s\end{pmatrix}^TH\begin{pmatrix}r\\ s\end{pmatrix},
\]
which implies that $(r,s)=(0,0)$ since $H$ is positive definite.
Therefore, $t=0$ which shows that 
$\nabla_{(\beta, \gam,v)}G_\emu((\hbeta, \hgam,\hv),(\tbeta,\tgam))$ is
nonsingular.

The formula for the inverse follows from \eqref{eq:grad G inv1} which tells us that
\begin{equation}\label{eq:grad G inv2}
\begin{bmatrix}H&A\\ B^T&D\end{bmatrix}^{-1}=
\begin{bmatrix}I&0\\ -B^TH^{-1}&I\end{bmatrix}
\begin{bmatrix}H^{-1}&0\\ 0&(\Diag{\hgam}+\Diag{\hv}\hat H_{2})^{-1}\end{bmatrix}
\begin{bmatrix}I&-H^{-1}A\\ 0&I\end{bmatrix}.
\end{equation}
Alternatively, one can apply the formulas in \cite{LS02}.
\end{proof}

Using Lemma \ref{lem:invert}, we
apply the implicit function theorem to the equation 
\[
G_\emu((\hbeta, \hgam,\hv),(\tbeta,\tgam))=0
\] 
and obtain the following result.

\begin{theorem}[Differentiability of $\SS_\emu$]\label{thm:diff sol map}
Let the hypotheses and notation of Lemma \ref{lem:invert} hold and let $\SS_\emu$
be as defined in \eqref{eq:argmin for u}. Given $\emu\in\R_+\times\R_+$,
define
$\map{\hSS_\emu}{\R^p\times\R^q}{\R^p \times\R^q_+\times\R^q_+}$ by
\begin{equation}\label{eq:hSS}
\hSS_\emu(\tbeta,\tgam)=
\lset{(\hbeta, \hgam,\hv)}{\hgam,\hv\in\R^q_+\ \text{and}\ 
0=G_\emu((\hbeta, \hgam,\hv),(\tbeta,\tgam))}
\end{equation}
Suppose
$(\tbeta,\tgam)\in\Rp\times\Rq$ and  
$(\bbeta, \bgam,\bv)=\hSS_\emu(\tbeta,\tgam)$
with $\bgam,\bv\in\R_+^q$ and such that \eqref{eq:str cs} holds. Then there
exist open neighborhoods $\widetilde\NN$ of $(\tbeta,\tgam)$ and 
such that 
$\SS_\emu$ and $\hSS_\emu$ are differentiable on $\widetilde\NN$ with
\[
\begin{aligned}
\nabla \hSS_\emu(\beta,\gam)&=\eta
\begin{bmatrix}
H^{-1}\!\!-\!\begin{bmatrix}\hat R\\ \hat H_2\end{bmatrix}\!
(D(\hgam)\!+\!D(\hv)\hat H_2)^{-1}D(\hv)[\hat R^T\ \hat H_2]
\\
\!-\,(D(\hgam)\!+\!D(\hv)\hat H_2)^{-1}D(\hv)[\hat R^T\ \hat H_2]
\end{bmatrix},
\\
\nabla \SS_\emu(\beta,\gam)&=\eta
\begin{bmatrix}
H^{-1}\!\!-\!\begin{bmatrix}\hat R\\ \hat H_2\end{bmatrix}\!
(D(\hgam)\!+\!D(\hv)\hat H_2)^{-1}D(\hv)[\hat R^T\ \hat H_2]
\end{bmatrix}
\end{aligned}
\]
for all $(\beta,\gam)\in \widetilde\NN$ and 
$(\hbeta,\hgam,\hv)=\hSS_\emu(\beta,\gam)$.
In particular, this implies that both $\hSS_\emu$ and $\SS_\emu$
are continuously differentiable on $\Rp\times\Rq$.
\end{theorem}

Using the notation of Lemma \ref{lem:invert} the expression for 
$\nabla \SS_\emu(\beta,\gam)$ in 
Theorem \ref{thm:diff sol map} can be simplified when $\mu>0$ to
{\small\[
\nabla \SS_\emu(\beta,\gam)\!=\eta
\!\left[ H^{-1}\!\!-\!
\begin{bmatrix}\!-\!H_1^{-1}R\\ I\end{bmatrix}
\hat H_2(\mu^{-1}\Diag{\hgam}^2\!+\!\hat H_2)^{-1}\!\hat H_2
[-R^TH_1^{-1}\ I]
\right]\! .
\]}
By combining this with the Shur complement formula
(e.g., see \eqref{eq:grad G inv1} and \eqref{eq:grad G inv2})
\[
H^{-1}=\begin{bmatrix}I&-H_1^{-1}R\\ 0&I\end{bmatrix}
\begin{bmatrix}H_1^{-1}&0\\ 0&(H_2-R^TH_1^{-1}R)^{-1}\end{bmatrix}
\begin{bmatrix}I&0\\ -R^TH_1^{-1}&I\end{bmatrix},
\]
where $\hat H_2=(H_2-R^TH_1^{-1}R)^{-1}$ is positive definite, we obtain  
{\small\[
\nabla \SS_\emu(\beta,\gam)\!=\!\eta
\begin{bmatrix}I&\!\!\!\!-H_1^{-1}R\\ 0&\!\!\!\!I\end{bmatrix}\!\!
\begin{bmatrix}H_1^{-1}&\!\!\!\!\!0\\ 
0&\!\!\!\!\!\hat H_2\!-\!\hat H_2(\mu^{-1}\Diag{\hgam}^2\!+\!\hat H_2)^{-1}\hat H_2
\end{bmatrix}\!\!
\begin{bmatrix}\!I&\!\!\!\!0\\ \!-R^TH_1^{-1}&\!\!\!\!I\end{bmatrix}
\]}
Since the matrix
\[
\hat H_2-\hat H_2(\mu^{-1}\Diag{\hgam}^2+\hat H_2)^{-1}\hat H_2=
\hat H_2^{1/2}[I-(I+\mu^{-1}\hat H_2^{-1/2}\Diag{\hgam}^2\hat H_2^{-1/2})^{-1}]
\hat H_2^{1/2}
\]
is positive definite, we have that
\begin{equation}\label{eq:grad S bound}
\norm{\nabla \SS_\emu(\beta,\gam)}\le
\eta(1+\norm{H_1^{-1}R}^2)
\max\{\norm{H_1^{-1}},\norm{(H_2-R^TH_1^{-1}R)^{-1}}\}.
\end{equation}
Since $H_1=\nabla_{\beta\beta}\LL(\hbeta,\hgam)+\eta I$, we have 
\(
\norm{H_1^{-1}}\le \eta^{-1}<(\eta-\bar\eta)^{-1}.
\)
We now show that  $(\eta-\bar\eta)^{-1}$ bounds
$\norm{(H_2-R^TH_1^{-1}R)^{-1}}$.
For this it is sufficient to show that $(\eta-\bar\eta)\le \mu_\mmin(H_2-R^TH_1^{-1}R)$.
By Lemma \ref{lem:LL weak cvx}, the matrix in \eqref{eq:psd for LL} is positive
semidefinite.
Since 
$\nabla_{\beta\beta}\LL(\beta,\gam)$ 
is positive
definite, the Shur complement 
\(
\nabla_{\gam\gam}\LL(\beta,\gam)+\blam I-
\nabla_{\beta\gam}\LL(\beta,\gam)\nabla_{\beta\beta}
\LL(\beta,\gam)^{-1}
\nabla_{\gam\beta}\LL(\beta,\gam)
\)
is positive semidefinite. Consequently,
\[
\begin{aligned}
H_2-R^TH_1^{-1}R=
(\eta-\bar\eta) I&+ (\nabla_{\gam\gam}\LL(\beta,\gam)+\blam I-
R^T\nabla_{\beta\beta}\LL(\beta,\gam)^{-1}R)
\\ &+R^T(\nabla_{\beta\beta}\LL(\beta,\gam)^{-1}-(\nabla_{\beta\beta}\LL(\beta,\gam)+\eta I)^{-1})R
\ \succeq (\eta-\bar\eta) I,
\end{aligned}
\]
since $\nabla_{\beta\beta}\LL(\beta,\gam)^{-1}-(\nabla_{\beta\beta}\LL(\beta,\gam)+\eta I)^{-1}$ is positive definite.
Therefore, $(\eta-\bar \eta)\le \mu_\mmin(H_2-R^TH_1^{-1}R)$.
By combining this with \eqref{eq:grad S bound} we obtain the bound
\begin{equation}\label{eq:S lip bd 1}
\norm{\nabla \SS_\emu(\beta,\gam)}\le
\frac{\eta}{\eta-\blam}\left(1+\norm{H_1^{-1}R}^2\right),
\end{equation}
where 
\[
\begin{aligned}
H_1^{-1}R&=-(X^T\Omega(\hgam)^{-1}X+\eta I)^{-1}
\sum_{i=1}^mX_i^T\Omega_i(\hgam)^{-1}Z_i
\Diag{Z_i^T\Omega_i(\hgam)^{-1}(X_i\hbeta-Y_i)}
\\ &=
-(X^T\Omega(\hgam)^{-1}X+\eta I)^{-1}X^T\Omega(\hgam)^{-1}
\hZ\,\Diag{\hZ^T\Omega(\hgam)^{-1}r(\hbeta)},
\end{aligned}
\]
with $r(\hbeta):=X\beta-y$ and 
\(
\hZ=\Diag{Z_1,Z_2,\dots,Z_m}.
\)
Therefore, as in Lemma \ref{lem:LL weak cvx}, we obtain the bound
\begin{equation}\label{eq:S lip bd 2}
\norm{H_1^{-1}R}\le \eta^{-1}\mu^{-2}_\mmin(\Lam)          
\sig_\mmax(X)\sig_\mmax^2(Z)\norm{X\hbeta -y}.
\end{equation}
This inequality can be used to show that $\nabla \uemu$ is bounded on the lower
level sets of $\uemu\tbg + R\tbg+\del_{\Rqp}(\tgam)$ if
$\LL_{\emu}(\bg,\tbg) + R\tbg+\del_{\R^q_+}(\tgam)$ is level
compact. However, we only know that this is true if we can bound the values 
of $\hgam$ over these sets. In practice, the values of $\hgam$ are bounded if the
model is well posed since these values are tied to the variances of
the random effects. One can accommodate this by adding a constraint of the 
form $\gam\le\gammax$ for $\gammax\in\R^q_{++}$ chosen sufficiently large.

\begin{lemma}[Lipschitz Continuity of $\nabla\uemu$]\label{lem:grad u lip}
Let the assumptions of Theorem \ref{thm:relaxed existence} hold and suppose
$\mu>0$ and $\gammax\in\R^q_{++}$. Let $\zeta\in\R$ and set 
\[
\begin{aligned}
\widehat\VV(\eta,\mu,\gammax,\zeta)&:=
\lset{(\bg,\tbg)}{\begin{array}{c}
\LL_{\emu}(\bg,\tbg) + R\tbg+\del_{\R^q_+}(\tgam)\le\zeta,\\ \gam,\ \tgam\le\gammax\end{array}},
\text{and}
\\
\VV(\eta,\mu,\gammax,\zeta)&:=\lset{\tbg}{\uemu\tbg + R\tbg+\del_{\Rqp}(\tgam)\le\zeta,\ \tgam\le\gammax}.
\end{aligned}
\]
Then 
\begin{enumerate}
\item
both $\widehat\VV(\eta,\mu,\gammax,\zeta)$ and 
$\VV(\eta,\mu,\gammax,\zeta)$ 
are compact with
\begin{equation}\label{eq:VV inclusion}
\VV(\eta,\mu,\gammax,\zeta)\subset\lset{\tbg}{\begin{aligned}\tgam\le\gammax
 \ \text{and}\ 
\exists\,\bg\in\Rp\times\R^q_{++}
\ s.t.
\\
(\bg,\tbg)\in\widehat\VV(\eta,\mu,\gammax,\zeta)
\end{aligned}},
\end{equation}
\item
the set $\VV(\eta,\mu,\gammax,\zeta)$ has nonempty interior if
$\zeta>\uemu\tbg$ for some $\tbg\in\Rp\times\Rq$, and
\item
the set
$\widetilde\VV(\eta,\mu,\gammax,\zeta,\omega):=\clco(\VV(\eta,\mu,\gammax,\zeta)+\omega\uball)$ is a compact, convex set with nonempty interior
whenever $\VV(\eta,\mu,\gammax,\zeta)\ne\emptyset$.
\end{enumerate}
Moreover, $\nabla u_\emu$ is Lipschitz on $\clco(\VV(\eta,\mu,\gammax,\zeta)+\omega\uball)$
for every $\omega\ge 0$, where 
\[
\uball:=\lset{\bg\in\Rp\times\R^q}{\norm{\bg}\le 1}.
\]
\end{lemma}
\begin{proof}
Since Theorem \ref{thm:basic existence} tells us that
$\LL$ is bounded below, 
$\LL_{\emu}(\bg,\tbg) + R\tbg+\del_{\R^q_+}(\tgam)$ is not level compact
if and only if there is an unbounded sequence
in a lower level set of 
$\LL_{\emu}(\bg,\tbg) + R\tbg+\del_{\R^q_+}(\tgam)$ 
for which $\gam^k\uparrow\infty$.
Therefore, the compactness of $\widehat\VV(\eta,\mu,\gammax,\zeta)$ 
follows from the lower semicontinuity of  
$\LL_{\emu}(\bg,\tbg) + R\tbg+\del_{\R^q_+}(\tgam)$.
Since
\[
\uemu\tbg + R\tbg+\del_{\Rqp}(\tgam)\le 
\LL_{\emu}(\bg,\tbg) + R\tbg+\del_{\R^q_+}(\tgam)\quad
\forall\, \bg\in\Rp\times\Rq,
\]
the inclusion \eqref{eq:VV inclusion} holds.
In addition, the set on the right hand side of \eqref{eq:VV inclusion}
is the projection of $\widehat\VV(\eta,\mu,\gammax,\zeta)$ onto
its
first components $\bg$ and so is compact. This in turns tells us that 
$\VV(\eta,\mu,\gammax,\zeta)$ is compact. Hence, 
$\clco(\VV(\eta,\mu,\gammax,\zeta)+\omega\uball)$
is also compact. The continuity of $\uemu$ implies that 
$\VV(\eta,\mu,\gammax,\zeta)$ has nonempty interior if
$\zeta>\uemu\tbg$ for some $\tbg\in\Rp\times\Rq$.
Theorem \ref{thm:u cont} shows that $\SS_\emu$ is continuous on 
$\Rp\times\Rq$ so the bound \eqref{eq:S lip bd 2} combined with 
Theorem \ref{thm:diff sol map} implies that $\SS_\emu$ is locally
Lipschitz on $\Rp\times\Rq$. Hence, by \eqref{eq:grad u},
$\nabla \uemu$ is locally Lipschitz on $\Rp\times\Rq$.
The compactness of 
$\clco(\VV(\eta,\mu,\gammax,\zeta)+\omega\uball)$
tells us that $\nabla \uemu$ is Lipschitz on
this set for all $\omega\ge 0$.
\end{proof}


\section{Convergence of the PGD Algorithm for $\Phi_\emu$}
\label{sec:convergence pgd}

The convergence of the PGD algorithm for fixed valued of the relaxation parameters
$\eta$ and $\mu$ appeals to the standard convergence theory as presented in
\cite[Chapter 10]{AB17} which requires the use of Assumptions (A)--(C) in Section 
\ref{sec:pgd}. We assume that the
variable selection regularizer $R$ is chosen so that
Assumption (A) holds. 
In addition, under the assumptions of Theorem 5, 
Theorem \ref{thm:u diff} tells us that the function $\uemu$ is well defined and continuously differentiable on all of $\Rp\times\Rq$
with the solution mapping 
$\SS_\emu\tbg$ well defined, single valued, and differentiable on $\Rp\times\Rq$ 
(Theorem \ref{thm:diff sol map}).
 Therefore, Assumption (C)
is satisfied as is much of assumption (B).
However, as is commonly the case in a specific application, the 
$L_\emu$-smoothness of $\uemu$ over 
$\intr{\dom{\uemu}}=\Rp\times\Rq$ fails. This drawback is remedied by
observing that the 
PGD algorithm is a descent algorithm. This allows us to focus on the behavior of
the functions over the lower level sets
described in Lemma \ref{lem:grad u lip}.

Let $\bw^{0}=(\tbeta^0,\tgam^0)\in \Rp\times\R^q_+$ be the point at which Algorithm \ref{alg:MSR3}
is initiated and let
$\zeta>\LL_{\emu}(\bg,(\tbeta^0,\tgam^0)) 
+ R(\tbeta^0,\tgam^0)+\del_{\R^q_+}(\tgam^0)$ for any 
$\bg \in \Rp\times\R^q_{++}$. 
For $\omega\ge 0$ and $\eps\ge 0$, define
\(
\Dfr(\omega,\eps):=\widetilde
\VV(\eta,\mu,\gammax+\eps\one,\zeta+\eps,\omega+\eps)
\)
and set
\begin{equation}\label{eq:hats}
\hat u_\emu:=\uemu+\delta_{\Dfr(\bomega,\beps)}\quad\text{ and }\quad
\hR:=R+\del_{\Dfr(\bomega,0)},
\end{equation}
where $\beps>0$, $\widetilde\VV$ is defined in Lemma \ref{lem:grad u lip} and
\[
\bomega:=1+t_0\,\max\lset{\nabla \uemu\tbg}{\tbg\in\VV(\eta,\mu,\gammax,\zeta)}.
\]
Observe that all iterates of Algorithm \ref{alg:MSR3} 
lie in the set $\VV(\eta,\mu,\gammax,\zeta)$
since it is a descent algorithm.
Moreover, since the prox operator is nonexpansive 
(e.g., see  
\cite[Theorem 6.42(a)]{AB17} or \cite[Theorem 12.19]{rockafellar2009variational}),
all of the points tested in the backtracking line search in Algorithm 
\ref{alg:MSR3} 
must also
lie in the set $\widetilde\VV(\eta,\mu,\gammax,\zeta,\bomega)$ by construction.
Therefore, the iterates of Algorithm \ref{alg:MSR3} 
are identical to those obtained
when the algorithm is applied to $\hat u_\emu$ with 
$\tR:=R+\del_{\Dfr(\bomega,0)}$.
That is, we can assume
that the Algorithm \ref{alg:MSR3} 
is being applied to $\hat u_\emu$. 
Observe that $\hat u_\emu$ is closed and proper,
$\dom{\hat u_\emu}=\Dfr(\bomega,\beps)$ is convex, and
 $\dom{\hat u_\emu}=\Dfr(\bomega,\beps)$
has nonempty interior (by Lemma \ref{lem:grad u lip}(3)) with $\dom{\tR}\subset \intr{\dom{\hat u_\emu}}$ since $\beps>0$. 
In addition, the final statement of Lemma \ref{lem:grad u lip}
tells us that there is an $L_{(\emu,\gammax,\zeta)}>0$ 
such that $\hat u_\emu$ is 
$L_{(\emu,\gammax)}$-smooth over $\intr{\dom{\hat u_\emu}}$. Hence, 
Assumptions (A)-(C) are satisfied by $\hat u_\emu$ and $\tR$ and so the 
convergence properties in \cite[Theorem 10.15]{AB17} hold for Algorithms
1 and 2 applied to $\uemu$ and $R$ under the assumptions of Theorem  
\ref{thm:relaxed existence}. By applying these observation
to \cite[Theorem 10.15]{AB17}, we obtain the following convergence result.

\begin{theorem}[Convergence of Algorithms \ref{alg:MSR3} and \ref{alg:pgd with bt}]\label{thm:convergence}
Let the assumptions of Theorem \ref{thm:relaxed existence} hold, and
let $\Phi_\emu$ be as defined in \eqref{eq:varphi}. 
Let $\{\tbgk\}$ be a sequence generated
by either Algorithm \ref{alg:MSR3} or \ref{alg:pgd with bt} 
with parameters 
$\theta\in(0,1),\ \tau\in(0,1)$,
$\eta>0,\ \mu>0$, $\eps_{\mbox{\tiny Tol}}=0$,
$t_0>0$, and $\gammax>\tgam^0$. 
Then, given $\zeta>\uemu(\tbeta^0,\tgam^0)$ there is an 
$L_{(\emu,\gammax,\zeta)}>0$ such that $\nabla\uemu$ is 
$L_{(\emu,\gammax,\zeta)}$-smooth over 
$\widetilde\VV((\emu,\gammax,\zeta,1)$. 
In Algorithm \ref{alg:pgd with bt}, replace 
$L_\emu$ with $L_{(\emu,\gammax,\zeta)}$ and set
 \[
 \begin{aligned}
 M&:=\begin{cases}
 \alf(1-\alf\frac{L_{(\emu,\gammax,\zeta)}}{2}),&
 \mbox{in Algorithm \ref{alg:pgd with bt}},
 \\
 \frac{2t_0\theta^2(1-\tau)}{\max\{2\theta(1-\tau),t_0L_{(\emu,\gammax,\zeta)}\}},&
 \mbox{in Algorithm \ref{alg:MSR3}},
 \end{cases}
 \qquad\text{and}
 \\
 r&:=
 \begin{cases}
 \alf,&
 \mbox{in Algorithm \ref{alg:pgd with bt}},
 \\
t_0,&
 \mbox{in Algorithm \ref{alg:MSR3}}.
 \end{cases}
 \end{aligned}\]
Then either $\tgam^k>\gammax$ after a finite number of iterations and the algorithms terminate, or the 
following hold:
\begin{enumerate}
\item
The sequence $\Phi_\emu\tbgk$ is nondecreasing.
In addition, \[\Phi_\emu(\tbgk)\tbeta^{k+1},\tgam^{k+1})<\Phi_\emu\tbgk\]
if and only if $\tbgk$ is not a stationary point of \eqref{eq:relax3}.
\item
$\norm{\tbgk-\prox_{r\tR}(\tbgk-r\nabla u_\emu\tbgk)}\rightarrow 0$ with
\[
\min_{i=0,1,\dots,k}\norm{\tbgk-\prox_{r\tR}(\tbgk-r\nabla u_\emu\tbgk)}
\le \frac{\sqrt{\Phi_\emu(\tbeta^0,\tgam^0)-\Phi_\emu^\mathsf{OPT}}}
{\sqrt{M(k+1)}},
\]
where $\Phi_\emu^\mathsf{OPT}:=\inf \Phi_\emu$.
\item
All limit points of the sequence $\{\tbgk\}$ are stationary points of
problem \eqref{eq:relax3}.
\end{enumerate}
\end{theorem}
\begin{proof}
As observed prior to the statement of the theorem, 
both Algorithm \ref{alg:MSR3} and \ref{alg:pgd with bt}
behave as if they were applied to the the functions 
$\hat u_\emu$ and $\hR$ defined in \eqref{eq:hats}. 
It was also shown that the functions $\hat u_\emu$ and $\hR$
satisfy the Assumptions (A)-(C) required by \cite[Theorem 10.15]{AB17}.
Hence, the consequences of \cite[Theorem 10.15]{AB17} hold.
By translating the notion of \cite[Theorem 10.15]{AB17} to that of this paper,
we obtain the result.
\end{proof}

\section{A Hybrid Algorithms for Feature Selection in Mixed Effects Models}
\label{sec:hybrid alg}

 In the previous section we established the convergence properties of the 
 PGD algorithm applied to the function $\Phi_\emu$ for fixed values of
 $\eta$ and $\mu$. In subsection \ref{ssec:existence and consistency},
 two consistency results are established for the relaxed problem \eqref{eq:relax3}.
 Theorem \ref{thm:eta consistency} shows that, for fixed $\mu\ge 0$,
 every limit point of solutions to \eqref{eq:relax3} as $\eta\uparrow\infty$
 is a solution to \eqref{eq:log-barrier problem}, while Theorem
 \ref{thm:mu consistency} tells us that every
 limit point of solutions \eqref{eq:log-barrier problem} as $\mu\downarrow 0$
 is a solution to the variable selection problem \eqref{eq:vs_lme}. 
These results suggest a range of numerical approachs to obtaining approximate 
solutions to the target problem \eqref{eq:vs_lme}.
 The issue of foremost concern is the method for approximating 
 solutions to \eqref{eq:value_function_definition} since the accuracy in this approximation
 determines the accuracy in both $\uemu$ and $\nabla \uemu$.
 To address this concern, we view the algorithm from an interior point
 perspective where every point on the {\it central path} is a solution
 to the optimization problem \eqref{eq:value_function_definition}
 defining $\uemu$ for the associated value of the homotopy parameter $\mu$. 
 An approximate solution is then
 considered acceptable if it is sufficiently close to the central path where
 proximity to the central path is measured in terms of the notion of the 
 {\it neighborhood} of the central path , e.g. see \cite{Wright-IP-book}. 
 Due to the convexity of the optimization problems \eqref{eq:value_function_definition}, 
 this is an efficient algorithm for approximating 
 $\uemu$ to high accuracy.

 The next issue we addressed is the method for initializing and adjusting the parameter 
 $\eta$.  This is particularly significant since the initial value of $\eta$ must be chosen
 to assure the convexity of the problems in \eqref{eq:value_function_definition}.
 Lemma \ref{lem:LL weak cvx} gives us guidance in this regard, but the necessary 
 computations to obtain a lower bound on $\eta$ can be arduous, and, in general,
 produce a wildly pessimistic lower bound. For this reason, we take a somewhat different 
 approach by proposing a variable metric strategy for solving the optimization problems
 in \eqref{eq:value_function_definition}. In this approach, we replace 
 the Hessian matrix $\nabla^2\LL_{\emu}\bg$ in
 the Newton equation
 
 \centerline{$
 G_\emu((\beta, \gamma, v), (\tbeta, \tgamma))+
 \nabla G_\emu((\beta, \gamma, v), (\tbeta, \tgamma))
 [dv, d\beta, d\gamma] =0
$}

\noindent
by the positive semi-definite approximation 

 \centerline{$
\nabla^2\LL{(\beta,\gam)} \approx \sum_{i=1}^m
S_i^T\begin{bmatrix}X_i^T\\ -Z_i^T\end{bmatrix}
\Omega_i(\gam)^{-1}
\begin{bmatrix}X_i& -Z_i\end{bmatrix}S_i
$}

\noindent
which is motivated by the expression for $\nabla^2\LL_{\emu}\bg$ given in
\eqref{eq:hess LL}. That is, we simply drop the negative semi-definite term
$-\sum_{i=1}^m\begin{bmatrix}
0&0\\ 0& \half(Z_i^T\Omega_i(\gam)^{-1}Z_i)^{\circ2}
\end{bmatrix}$.
With this modification, the subproblems we solve are strongly convex for all $\eta>0$.
Consequently, the problem of initializing $\eta$ is less problematic. Our numerical experiments indicate that the performance of the algorithm is robust with respect to 
$\eta$. For this reason, we choose an initial value for $\eta$ and 
then leave it fixed over all iterations. Our method for choosing $\eta$ is described
in \cite[Section 4, Figure 5]{Practice}. Briefly, we maximize the 
Baysian Information Criterion (BIC) over a grid of values for $\eta$. 
 The resulting BIC response curve shows 
that the method is robust with respect to the 
 choice of $\eta$ and choosing $\eta\in [1,10]$ yields 
 accurate solutions for our selected test problems. 
 Once $\eta$ is fixed the
 PGD algorithm can be applied to solve the problem \eqref{eq:relax3}
 for decreasing values of $\mu$. 

 Finally, we propose two methods for updating $\mu$. In the first,  
$\mu$ is reduced by a fixed percentage of its current estimate 
after obtaining an approximate
solution to the equation \eqref{eq:kkt}, i.e., an approximate 
KKT point for the optimization problem defining $\uemu$. We call this
method MSR3 (Algortihm \ref{alg:MSR3}).
In the second, we update $\mu$ after each interior point iteration lying
in a neighborhood of the central path, and call this 
more aggressive algorithm MSR3-fast 
(Algorithm \ref{alg:MSR3-fast}).
\medskip

{\small
\begin{algorithm}[H]
\SetAlgoLined
$\texttt{progress}\leftarrow \textbf{True}$;
$\beta^+, \tbeta^+\leftarrow\beta_0$; \quad $\gamma^+, \tgamma^+\leftarrow\gamma_0$; \quad \texttt{outer\_iter = 0} \\
\While{\texttt{outer\_iter} $<$ \texttt{max\_iter}  \ and \ $\texttt{progress}$\\}
{
$v^+ \leftarrow 1 \in \R^q$; \quad $\mu \leftarrow \frac{{v^+}^T\gamma^+}{10 q}$; \quad \texttt{inner\_iter = 0}\\
 \While{\texttt{inner\_iter} $<$ \texttt{max\_iter}  \ and \ $\|G_\emu(\beta^+, \gamma^+, v^+)\|$ $>$ \texttt{tol}   \ and  \ ($\|\beta^+ - \beta\| \geq \text{tol}$ or $\|\gamma^+ - \gamma\|  \geq \text{tol}$)\\}{
    $\beta \leftarrow \beta^+$; \quad $\gamma \leftarrow \gamma^+$; \quad $\tbeta \leftarrow \tbeta^+$; \quad $\tgamma \leftarrow \tgamma^+$ \\
    $[dv, d\beta, d\gamma] \leftarrow  \nabla G_\emu((\beta, \gamma, v), (\tbeta, \tgamma))^{-1}  G_\emu((\beta, \gamma, v), (\tbeta, \tgamma))$ \tcp*[f]{Newton Iteration}\\ 
    $\alpha \leftarrow 0.99\times\min\left(1, -\frac{\gamma_i}{d\gamma_i}, \forall i :\ d\gamma_i < 0\right)$\\
    $\beta^+ \leftarrow \beta + \alpha d\beta$; \quad $\gamma^+ = \gamma + \alpha d\gamma$; \quad  $v^+ \leftarrow v + \alpha dv$\\
    \If{$\|\gamma^+\odot v^+ - q^{-1}{\gamma^+}^Tv^+ \mathbf{1}\| > 0.5q^{-1}{v^+}^T\gamma^+$}{
    	continue \tcp*[f]{Keep doing Newton iterations}\\
    }
    \Else{ 
        $\mu \leftarrow \frac{{v^+}^T\gamma^+}{10 q}$ \tcp*[f]{Near central path}\\
    }
    \texttt{inner\_iter += 1}
 }
 $\tbeta^+ = \prox_{\alpha R}(\beta^+)$; \quad $\tgamma^+ = \prox_{\alpha R + \delta_{\R_+}}(\gamma^+)$; \quad \texttt{outer\_iter += 1}; \quad
$ \texttt{progress}=$
($\|\tbeta^+ - \tbeta\| \geq \text{tol}$ or $\|\tgamma^+ - \tgamma\| \geq \text{tol}$) 
 }
 \Return{$\tbeta^+$, $\tgamma^+$}
 \caption{\label{alg:MSR3}MSR3 (Hybrid PGD for LME feature selection)}
\end{algorithm}
}

{\small
\begin{algorithm}[H]
\SetAlgoLined
$\texttt{progress}\leftarrow \textbf{True}$; \quad \texttt{iter = 0}; \\
$\beta^+, \tbeta^+\leftarrow\beta_0$; 
\quad $\gamma^+, \tgamma^+\leftarrow\gamma_0$;  
\quad $v^+ \leftarrow 1 \in \R^q$; 
\quad  $\mu \leftarrow \frac{{v^+}^T\gamma^+}{10 q}$\\
 \While{\texttt{iter} $<$ \texttt{max\_iter}  \ and \ $\|G_\emu(\beta^+, \gamma^+, v^+)\|$ $>$ \texttt{tol}   \ and  \ \texttt{progress} \\}{
    $\beta \leftarrow \beta^+$; \quad $\gamma \leftarrow \gamma^+$; \quad $\tbeta \leftarrow \tbeta^+$; \quad $\tgamma \leftarrow \tgamma^+$ \\
    $[dv, d\beta, d\gamma] \leftarrow  \nabla G_\emu((\beta, \gamma, v), (\tbeta, \tgamma))^{-1}  G_\emu((\beta, \gamma, v), (\tbeta, \tgamma))$ 
 \quad   {\small\bf // Newton Iteration}\\ 
    $\alpha \leftarrow 0.99\times\min\left(1, -\frac{\gamma_i}{d\gamma_i}, \forall i :\ d\gamma_i < 0\right)$\\
    $\beta^+ \leftarrow \beta + \alpha d\beta$; \quad $\gamma^+ = \gamma + \alpha d\gamma$; \quad  $v^+ \leftarrow v + \alpha dv$\\
    \If{$\|\gamma^+\odot v^+ - q^{-1}{\gamma^+}^Tv^+ \mathbf{1}\| > 0.5q^{-1}{v^+}^T\gamma^+$}{
    	continue \tcp*[f]{Keep doing Newton iterations}\\
    }
    \Else{ 
        $\tbeta^+ = \prox_{\alpha R}(\beta^+)$;
        \    $\tgamma^+ = \prox_{\alpha R + \delta_{\R_+}}(\gamma^+)$; 
        \    $\mu = \frac{1}{10}\frac{{v^+}^T\gamma^+}{q}$ \tcp*[f]{Near central path} 
    }
    \texttt{progress} = ($\|\beta^+ - \beta\| \geq \text{tol}$ or $\|\gamma^+ - \gamma\|  \geq \text{tol}$ or $\|\tbeta^+ - \tbeta\| \geq \text{tol}$ or $\|\tgamma^+ - \tgamma\| \geq \text{tol}$)\\
    \texttt{iter += 1}
 }
 \Return{$\tbeta^+$, $\tgamma^+$}
 \caption{\label{alg:MSR3-fast}MSR3-fast (Accelerated hybrid PGD for LME feature selection)}
\end{algorithm}
}
\medskip




\section{Numerical results}
\label{sec:numerical}

A detailed numerical study and comparison Algorithms
\ref{alg:pgd_for_lme}, \ref{alg:MSR3}, and \ref{alg:MSR3-fast} as well as other algorithms for 
variable selection in LME models is given in \cite{Practice}.
Here we give one illustration from \cite{Practice}.
\smallskip

\noindent
{\bf Experimental Setup.} In this experiment we take the number of fixed effects $p$ and random effects $q$ to be $20$. We set $\beta = \gamma = [\frac{1}{2}, \frac{2}{2}, \frac{3}{2}, \dots, \frac{10}{2}, 0, 0, 0, \dots, 0]$, i.e. the first 10 covariates are increasingly important and the last 10 covariates are not. The data is generated as 
\[
\begin{aligned}
y_i &= X_i\beta + Z_iu_i + \varepsilon_i, \quad  \varepsilon_i \sim \NN(0, 0.3^2 I) \\
X_i &\sim \NN(0, I)^p, \quad Z_i = X_i \\ 
u_i& \sim \NN(0, \Diag{\gamma})\\ 
\end{aligned}
\]
We generated 9 groups with the sizes of $[10, 15, 4, 8, 3, 5, 18, 9, 6]$ to capture a variety of group sizes. To estimate the uncertainty bounds, each experiment is repeated 100 times.
\medskip

Table \ref{table:comparison_of_algorithms} compares the performance
of algorithms \ref{alg:pgd_for_lme}, \ref{alg:MSR3}, and \ref{alg:MSR3-fast}
for four different feature selection regularizers: L0 (the $\ell_0$-norm), 
L1 (the $\ell_1$-norm), ALASSO (adaptive LASSO \cite{Krishna2008,Lan2006, Xu2015,Lin2013, fan2014robust, pan2018simultaneous}), and SCAD 
(smoothed clipped absolute deviation \cite{Fan2012,chen2015inference,Ghosh2018}). Figure 
\ref{fig:bullying_data_fixed_feature_selection} gives a more detailed
picture of the statistical performance of the algorithms over the set
of 100 test problems.
The L1 and ALASSO regularizers are convex while the L0 and SCAD are not.
Despite the non-convexity of the L0 and SCAD regularizers, they exhibit 
superior accuracy in identifying the correct features.  
There are closed form expressions for the prox operator for all of these
regularizers \cite{Practice}. The hybrid MSR3-fast Algorithm 
\ref{alg:MSR3-fast}
is the clear winner
in terms of efficiency in that it produces highly accurate solutions in a 
tiny fraction of the time it takes Algorithms 
\ref{alg:pgd_for_lme} and \ref{alg:MSR3}.
As expected, the vanilla PGD Algorithm \ref{alg:pgd_for_lme} is the least accurate
in identifying the correct features since it is only a first-order method while
Algorithms \ref{alg:MSR3} and \ref{alg:MSR3-fast} both use higher-order 
information as well as incorporating global variational information on the
relaxed objectives $\LL_{\eta,\mu}$. The whisker plots
in Figure \ref{fig:bullying_data_fixed_feature_selection} show that although
Algorithm \ref{alg:MSR3} has a slight edge in accuracy, 
Algorithm \ref{alg:MSR3-fast} strongly
dominates both Algorithms \ref{alg:pgd_for_lme} and \ref{alg:MSR3}
in speed.

\begin{table}
\centering
\begin{tabular}{lllll}
\toprule
     & Model &    PGD &    MSR3 & MSR3-fast \\
Regularizer & Metric &        &        &       \\
\midrule
L0 & Accuracy &   0.89 &   \textbf{0.92} &  \textbf{0.92} \\
     & Time &  41.68 &  88.54 &  \textbf{0.13} \\
L1 & Accuracy &   0.73 &   \textbf{0.88} &  \textbf{0.88} \\
     & Time &  38.39 &   9.13 &  \textbf{0.13} \\
ALASSO & Accuracy &   0.88 &   \textbf{0.92} &  0.91 \\
     & Time &  34.55 &  65.19 &  \textbf{0.12} \\
SCAD & Accuracy &   0.71 &   \textbf{0.93} &  0.92 \\
     & Time &  77.62 &  84.67 &  \textbf{0.17} \\
\bottomrule
\end{tabular}

\caption{\label{table:comparison_of_algorithms} \small
Comparison of performance of algorithms measured as accuracy of selecting the correct covariates and run-time. The L0 strategy stands out 
over other standard regularizers. MSR3 improves performance significantly for all regularizers, while MSR3-fast improves convergence speed while preserving the 
accuracy of MSR3.  
}
\end{table}

\begin{figure}
    \centering
	\includegraphics[width=1.0\textwidth]{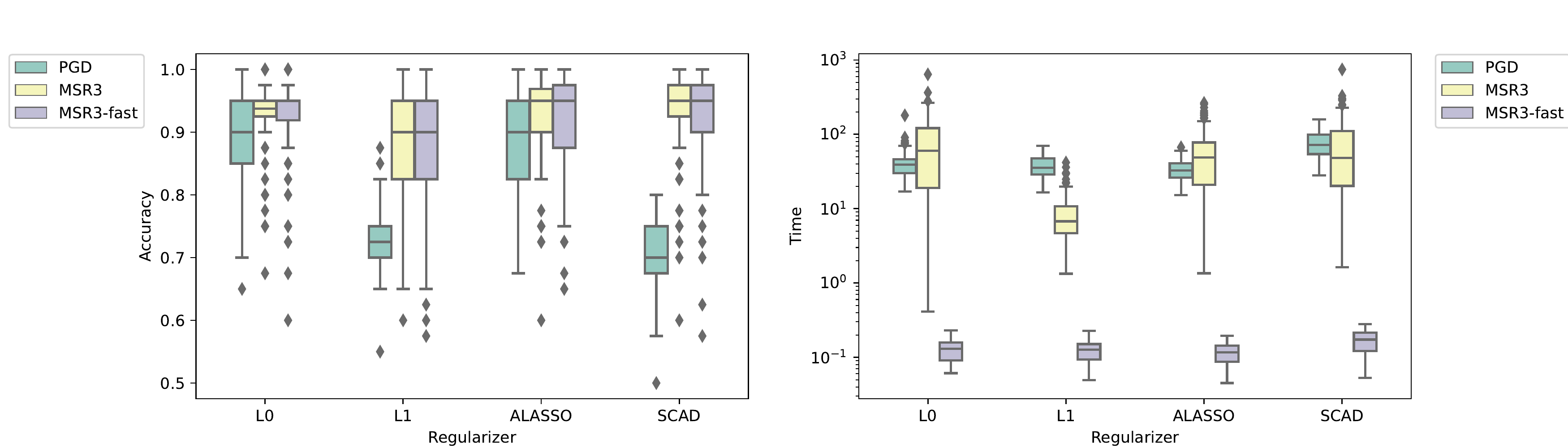}
	\caption{\label{fig:bullying_data_fixed_feature_selection} \small
	Feature selection accuracy and execution time in seconds for PGD and MSR3 with various regularizers. MSR3-Fast has the same accuracy as 
	MSR3 and significantly decreases computation time.}
\end{figure}

\bibliographystyle{amsplain}
\bibliography{bibliography}


\appendix

\section{Existence of Minimizers (Theorems~\ref{thm:basic existence} and~\ref{thm:basic existence2})}\label{adx:basic existence}

The key tool to prove existence of minimizers for both the likelihood and the penalized likelihood 
is the function $\map{f}{\R^n\times\bS_{++}^n}{\R}$ given by
\eq{\label{eq:key1}
f(r,M):=\half[r^TM^{-1}r+\ln |M|]\ .
}
If $M=U\Diag{\mu}U^T$ is the eigenvalue decomposition for $M$ where
$U^TU=I$, and $\tr=U^Tr$, then 
\eq{\label{eq:key2}
f(r,M)=\half\left[\sum_{i=1}^n \frac{\tr_i^2}{\mu_i}+\ln(\mu_i)\right]. 
}
For $\rho>0$, observe that
\(
\frac{\rho^2}{\omega}+\ln(\omega)
\) 
is greater that both $\ln(\omega)$ and $1+2\ln(\rho)$ for all $\omega>0$.
Therefore, using the facts $\mu_\mmax(M)=\norm{M}$ and $\norm{\tr}_\infty\ge(\norm{\tr}/\sqrt{n})
=(\norm{r}/\sqrt{n})$, we have
\eq{\label{eq:eigbd}
f(r,M)&\ge\half\sum_{i=1}^n \max\{\ln\mu_i,\,1+2\ln|\tr_i|\}
\\ &\ge
\max\{1\!+\!2\ln(\norm{r}/\sqrt{n})\!+\!\frac{n\!-\!1}{2}\ln\mu_\mmin(M),\
\ln\norm{M}\!+\!\frac{n\!-\!1}{2}\ln\mu_\mmin(M)\}
\\ &
\ge\max\{\ln(\norm{r}^2/n),\ln\norm{M}\} + \frac{n\!-\!1}{2}\ln\mu_\mmin(M),
}
where $\mu_\mmin(M)$ and $\mu_\mmax(M)$ are the smallest
and largest eigenvalue of $M$, respectively.
We have the following result due to \cite{zheng2021trimmed}
modified slightly with an independent proof. 
\begin{lemma}[Level Compactness of $f$]
\label{lem:levelcompact1}\cite[Theorem 1]{zheng2021trimmed}
Let $f$ be as given in \eqref{eq:key1}. Then, given $\rho\in\R$ and $\alf>0$, the set 
\[
\DD_{\rho,\alf}:=
\lset{(r,M)\in \R^n\times \bS_{++}^n}{f(r,M)\le \rho\ \text{and}\
\mu_\mmin(M)\ge\alf}
\] 
is compact, where $\mu_\mmin(M)$ and $\mu_\mmax(M)$ are the smallest
and largest eigenvalue of $M$, respectively.
\end{lemma}
\begin{proof}
If $\DD_{\rho,\alf}=\emptyset$, it is compact so we assume it is not empty.
Since $f$ is continuous on $\DD_{\rho,\alf}$, we need only show that this set 
is bounded. The boundedness of this set follows immediately from \eqref{eq:eigbd}. Indeed, if
$\{(r^k,M_k)\}\subset \R^n\times \bS_{++}^n$ 
diverges in norm then,  
without loss of generality, either $\norm{r^k}\rightarrow \infty$
or $\mu_\mmax(M)=\norm{M_k}\rightarrow\infty$, or both in which case 
\eqref{eq:eigbd} tells us that $f(r^k,M_k)\rightarrow\infty$.
\end{proof}
Observe that
\[
\LL_{ML}(\beta,\Gam)=f( r(\beta), \Omega(\Gam)) 
\]
where $\map{r}{\R^p}{n}$ and
$\map{\Omega}{\R^q}{\bS^n}$ are the affine transformations
\[\begin{aligned}
 r(\beta)&:=X\beta-Y,\qquad\text{and}
\\
\Omega(\Gam)&:=
\Diag{\Lambda_1+Z_1\Gam Z_1^T,\dots,
\Lambda_m+Z_m\Gam Z_m^T}.
\end{aligned}\]
For 
$i=1,\dots,m$, define
\[
\omega^i_\mmin:=
\mu_{\text{min}}(\Lambda_i)+\mu_\mmin(\Gam) 
\sig^2_{\text{min}}(Z_i)\quad \mbox{ and }
\quad
\omega_\mmin:=
\min_{i=1,\dots,m}\omega^i_\mmin,
\]
where $\mu_{\text{min}}(\Psi)$ 
and $\sig_{\text{min}}(\Phi)$ 
are the smallest eigenvalues and singular-values of 
$\Psi$ and $\Phi$,
respectively. 
By \cite[Theorem 3.1]{ABBP2021},
\eq{\label{eq:eig1}
0<\talf:=\mu_{\text{min}}(\Lambda)\le \omega_{\text{min}}
\le \mu_{\text{min}}(\Omega(\Gam))\quad \forall\ \Gam\in\bS_{+}^{q}.
}

\paragraph{\bf Proof for Theorem~\ref{thm:basic existence}}

The bound \eqref{eq:eig1}  
tell us that 
\[
\hat\DD:=\lset{(r,\Omega(\Gam))}{r\in\R^n,\, \Gam\in\bS_{+}^{q}
\ \text{and}\ f(r,\Omega(\Gam))\le\rho}
\subset\DD_{\talf,\rho}.
\]
In particular, 
$\LL$ is bounded below by \eqref{eq:eigbd}.
Hence there exists a sequence 
$\{(\beta^k,\Gam^k)\}\subset\R^p\times\bS_{+}^{q}$
such that 
\[
\LL_{ML}(\beta^k,\Gam_k)\downarrow
\inf_{\beta\in\R^p,\Gam\in\bS_{+}^{q}}\LL_{ML}(\beta,\Gam).
\]
Let $\rho=\LL_{ML}(\beta^0,\Gam_0)$.
Since $f$ is continuous on $\hat\DD\subset\DD_{\talf,\rho}$, 
$\DD_{\talf,\rho}$ is compact by Lemma \ref{lem:levelcompact1},
and both $\im{X}$ and $\im{\Omega}$ are closed,
%
%
with no loss in generality there is a
$(\bxi,\overline\Omega)\in\im{r}\times\im{\Omega}
\cap \DD_{\talf,\rho}$ 
such that
$(r(\beta^k),\Omega(\Gam_k))\rightarrow(\bxi,\overline\Omega)$. 
Since $(\bxi,\overline\Omega)\in \im{r}\times\im{\Omega}$,
there is a $(\bbeta,\overline\Gam)\in\R^p\times\bS^q_+$ such that
$(\bxi,\overline\Omega)=(r(\bbeta),\Omega(\overline\Gam))$.
In addition, since
$0<\talf\le \mu_\mmin(\Omega(\Gam))$ for all
$\Gam\in\bS_{+}^{q}$, we have $\LL_{ML}$ is lsc at $(\bbeta,\overline\Gam)$ 
telling us that $\LL(\bbeta,\overline\Gam)
=\inf_{\beta\in\R^p,\Gam\in\bS_{+}^{q}}\LL(\beta,\Gam)$.

\paragraph{\bf Proof of Theorem~\ref{thm:basic existence2}}

Define the affine transformations
$\map{\hOmega}{\R^q}{\bS^n}$ and $\map{\hOmega_i}{\R^q}{\bS^{n_i}}$ by
\begin{equation}
\hOmega(\gam):=\Omega(\Diag{\gam})\ \ \text{ and }\ \ 
\hOmega_i(\gam):=\Omega_i(\Diag{\gam})\quad i=1,\dots,m.
\end{equation}

The existence of a solution follows immediately once the level compactness
of $\L+\hR$ is establinshed. To this end
observe that
$\LL(\beta,\gam)=\LL_{ML}(\beta, \Diag{\gamma})=f( r(\beta), \hOmega(\gamma))$ 
and so 
\eqref{eq:eigbd} and  \eqref{eq:eig1} tell us that
\(
\LL\bg\ge  \frac{n\!+\!1}{2}\ln\talf.
\)
Since $\hR$ is level compact, it is lower bounded.
Therefore, $\LL +\hR$ is bounded below.
Let $\rho\in\R$ and $\{\bgk\}
\subset
\lset{\bg}{\LL\bg +\hR\bg\le\rho}$. 
We need to show that
$\{\bgk\}$ is bounded.
If $\norm{\bgk}\rightarrow \infty$,
then $\hR\bgk\rightarrow \infty$. Since $\LL\bgk +\hR\bgk\le\rho$, we must
have $\LL\bgk\rightarrow -\infty$. But $\LL$ is bounded below, hence 
$\{\bgk\}$must be bounded, and so $\LL +\hR$ is level compact.

\section{Derivatives of $\LL$}
\label{appendix:derivatives_of_lmm}
One can show that
\eq{
	\label{eq:all_derivatives}
\nabla_\beta\LL(\beta, \gamma) & = \sum_{i=1}^m X_i^T\Omega_i^{-1}(X_i\beta-Y_i)
=X^T\Omega^{-1}  r(\beta)
	\\
\nabla_\gamma\LL(\beta, \gamma) & = 
\half \sum_{i = 1}^m  \diag{{Z_i}^T\Omega_i^{-1}Z_i} - ({Z_i}^T\Omega_i^{-1}(X_i\beta-Y_i))^{\circ 2}
    \\
\nabla_{\beta\beta}\LL(\beta, \gamma) 
& =  
\sum_{i = 1}^m X_i^T\Omega_i^{-1}X_i 
=X^T\Omega^{-1} X
	\\
\nabla_{\beta\gamma}\LL(\beta, \gamma) & = 
-\sum_{i = 1}^m \Diag{Z_i^T\Omega_i^{-1}
(X_i\beta-Y_i)}Z_i^T\Omega_i^{-1}X_i 
	\\
\nabla_{\gamma\gamma}\LL(\beta, \gamma) 
& 
	\!=\!\!\sum_{i = 1}^m\! \left(Z_i^T\Omega_i^{-1}
(X_i\beta\!-\!Y_i)(X_i\beta\!-\!Y_i)^T
	\Omega_i^{-1}Z_i\right)
	\!\circ\!(Z_i^T\Omega_i^{-1}Z_i)
	\!-\!\mbox{\small$\half$}(Z_i^T\Omega_i^{-1}Z_i)^{\circ 2}
,}
where the final representation can be rewritten
using the fact that
$(yz^T)\circ A=\Diag{y}A\Diag{z}$ for all
$ y\in\R^m,z\in\R^n,\ A\in\R^{m\times n}$
\cite{HJ85}.
These formulas 
imply that
\begin{equation}\label{eq:hess LL}
\nabla^2\LL{(\beta,\gam)}=\sum_{i=1}^m
S_i^T\begin{bmatrix}X_i^T\\ -Z_i^T\end{bmatrix}
\Omega_i(\gam)^{-1}
\begin{bmatrix}X_i^T& -Z_i^T\end{bmatrix}S_i
-\begin{bmatrix}
0&0\\ 0& \half(Z_i^T\Omega_i(\gam)^{-1}Z_i)^{\circ2}
\end{bmatrix},
\end{equation}
where 
\[
S_i:=\begin{bmatrix}
I_q&0\\ 0&\Diag{Z_i^T\Omega_i(\gam)^{-1}(X_i\beta-Y_i)}
\end{bmatrix}.
\]
These formulas are also derived in \cite{Lindstrom1988}.

\end{document}